\documentclass[lettersize,journal]{IEEEtran}

\usepackage{amsmath, amsthm, amssymb, amsfonts, bm, bbm, mathrsfs, dsfont, breqn}
\allowdisplaybreaks

\usepackage{newtxtext}  % 推荐搭配，涵盖数学和正文，强烈建议这样写

\usepackage{graphicx}
\usepackage[caption=false]{subfig}
\usepackage{float}
\usepackage{multirow}
\usepackage{tabularx}

\usepackage{algorithm}
\usepackage{algpseudocode}

\usepackage{textcomp}
\usepackage{stfloats}
\usepackage{url}
\usepackage{verbatim}
\usepackage{cite}
\usepackage{xcolor}

\newtheorem{theorem}{Theorem}
\newtheorem{lemma}{Lemma}

\newtheorem{assumption}{Assumption}

\begin{document}

\title{FedOC: Multi-Server FL with Overlapping Client Relays in Wireless Edge Networks}

\author{Yun Ji, Zeyu Chen, Xiaoxiong Zhong,~\IEEEmembership{Member,~IEEE}, Yanan Ma, Sheng Zhang, and \\Yuguang Fang,~\IEEEmembership{Fellow,~IEEE} \thanks{Yun Ji, zeyu Chen and Sheng Zhang are with the Key Laboratory of Advanced Sensor and Integrated System, the Tsinghua Shenzhen International Graduate School, Tsinghua University, Shenzhen 518055, China. (e-mails: jiyunthu@gmail.com; zy-chen23@mails.tsinghua.edu.cn, zhang\_sh@mail.tsinghua.edu.cn).}\thanks{Xiaoxiong Zhong is with the Department of New Networks, Peng Cheng Laboratory, Shenzhen 518000, China, (e-mail: xixzhong@gmail.com).}\thanks{Yanan Ma and Yuguang Fang are with the Hong Kong JC STEM Lab of Smart City and the Department of Computer Science, City University of Hong Kong, Kowloon, Hong Kong. (e-mail: yananma8-c@my.cityu.edu.hk; my.Fang@cityu.edu.hk). \textit{(Corresponding authors: Sheng Zhang and Xiaoxiong Zhong)}}}

\maketitle

\begin{abstract}
Multi-server Federated Learning (FL) has emerged as a promising solution to mitigate communication bottlenecks of single-server FL. We focus on a typical multi-server FL architecture, where the regions covered by different edge servers (ESs) may overlap. A key observation of this architecture is that clients located in the overlapping areas can access edge models from multiple ESs. Building on this insight, we propose FedOC (Federated learning with Overlapping Clients), a novel framework designed to fully exploit the potential of these overlapping clients. In FedOC, overlapping clients could serve dual roles: (1) as Relay Overlapping Clients (ROCs), they forward edge models between neighboring ESs in real time to facilitate model sharing among different ESs; and (2) as Normal Overlapping Clients (NOCs), they dynamically select their initial model for local training based on the edge model delivery time, which enables indirect data fusion among different regions of ESs. The overall FedOC workflow proceeds as follows: in every round, each client trains local model based on the earliest received edge model and transmits to the respective ESs for model aggregation. Then each ES transmits the aggregated edge model to neighboring ESs through ROC relaying. Upon receiving the relayed models, each ES performs a second aggregation and subsequently broadcasts the updated model to covered clients. The existence of ROCs enables the model of each ES to be disseminated to the other ESs in a decentralized manner, which indirectly achieves inter-cell model and speeding up the training process, making it well-suited for latency-sensitive edge environments. Extensive experimental results show remarkable performance gains of our scheme compared to existing methods.

\end{abstract}


\section{Introduction}

\IEEEPARstart{O}{ver} the last decade, Machine Learning (ML) has driven significant advances in data-driven applications such as personalized medicine, autonomous driving, and urban-scale sensing~\cite{a1,a2}. These achievements have been largely fueled by the abundance of data available for model training. However, privacy concerns often limit centralized data collection, particularly when sensitive information such as personal images or medical records is involved. Federated Learning (FL) \cite{a3} has emerged to address this challenge by enabling a shared model to be trained across data distributed on multiple clients, thereby achieving scalable learning while preserving user data privacy. 

A typical FL system operates in rounds, where clients perform local training based on their own datasets and upload model updates to a cloud server for aggregation and broadcast. Through repeated iterations, a global model is gradually refined~\cite{a3}. Since FL is trained on multiple clients with distinct data distributions and varying computation resources, the straggler problem and data heterogeneity are two critical challenges~\cite{a10,a11}, and have attracted substantial research attention \cite{a4,a5,a6}. Although promising progresses have been made, an inherent limitation of this single-server architecture is the long-range communication latency between the client and the cloud server (CS) in large-scale FL systems. With the growing demand for latency-sensitive applications such as autonomous driving and UAV-based disaster response, FL architectures incorporating multiple ESs have been proposed to alleviate communication delays between clients and the CS.

A typical multi-server FL scheme is Hierarchical FL (HFL) ~\cite{a7,a8,a12,a14,a15,a16}, where ESs are responsible for aggregating local model updates from their associated clients and subsequently uploading the aggregated model to the CS. However, frequent communications between ESs and the CS can still lead to high latency and resource consumption~\cite{a13}. Another widely adopted multi-server FL mechanism is Clustered Federated Learning (CFL)~\cite{a36,a37,a38,a39}, where clients are dynamically grouped into clusters based on data or model similarity, and each cluster trains a separate model tailored to its members. However, CFL often requires repeated re-clustering across communication rounds to adapt to client dynamics, which can incur significant communication and computation overhead. 

In addition, most existing multi-server FL methods are based on the assumption of disjoint edge server coverage, which may not hold in some realistic deployments of 5G and beyond systems, where a client may fall within the overlapping coverage of multiple ESs~\cite{a17}. In this paper, we consider such practical scenarios, where the overlapping coverage of edge servers gives rise to a new category of clients—overlapping clients (OCs). Each OC can communicate with multiple ESs within the same FL round. We have the following key observations about the potential roles of OCs. From a communication perspective, OCs can act as relays by forwarding models downloaded from one ES to other reachable ESs during the FL process. From a computational perspective, OCs can dynamically select the ES model as the starting point for local training according to the reception time of the edge models, which can mitigate the data heterogeneity across different cells. In essence, OCs have the potential to facilitate distributed model fusion across cells without requiring frequent CS aggregations, thereby reducing overall communication latency.

To our best knowledge, utilizing OCs to accelerate FL training has seldom been studied. The most closely related works are \cite{a17} and \cite{a18}. In \cite{a17}, Han \textit{et al.} proposed FedMES, a distributed FL scheme that allows OCs to train the local models based on the average of received edge models from multiple ESs. However, due to heterogeneous client conditions across edge regions, OCs must wait for all relevant ESs to deliver their models before training, leading to considerable delay when inter-edge heterogeneity is pronounced. Moreover, when the number of OCs is limited, their contribution to global model convergence is marginal. In \cite{a18}, OCs store all received models from multiple ESs and upload them together after completing local training in the next round. However, this approach introduces stale model sharing and increases the storage burden on clients. In contrast, our framework enables real-time model relaying between ESs and incorporates latency-aware strategies for OCs to adaptively participate training. In summary, we highlight our main contributions as follows:
\begin{itemize}
    \item We propose FedOC, a novel multi-server FL framework that leverages OCs to facilitate cross-server collaboration. In FedOC, OCs can act as relay nodes that forward edge models between adjacent cells upon reception, thereby enabling distributed aggregation across all ESs. To reduce relay communication overhead, we assign a single relay OC per overlapping region and merge model relaying with local update uploading into a single transmission. We also propose a lightweight scheduling scheme that uses orthogonal frequency resources to avoid inter-cell interference.
    
    \item To alleviate inter-cell data heterogeneity and reduce training latency, we develop latency-aware training strategies for OCs, where each OC selects the earliest-received edge model as training start in each training round. For comparison, we also define a fixed selection strategy for OCs, where each overlapping client consistently uses the model from a designated home ES, behaving as a regular local client of that server.

    \item We provide both theoretical convergence analysis and extensive empirical evaluation of FedOC under heterogeneous data settings. Results show that FedOC achieves faster convergence and higher training efficiency compared to existing benchmarks.
\end{itemize}
The remainder of this paper is organized as follows. Section II introduces the related work. Section III presents the system architecture and federated learning training process. Section IV provides theoretical results on the convergence of FedOC. Section V evaluates the performance of the proposed framework with simulations. We conclude the paper in Section VI.

\section{Related Works}
Since the proposal of FL \cite{a3}, significant research efforts have been devoted to exploring its implementation over wireless networks. To reduce the training time to achieve a desired learning accuracy, many works have focused on optimizing client selection and radio resource management. For instance, joint client selection and bandwidth allocation policies were introduced in \cite{a5, a6, a19, a20,a21} to minimize the learning loss or maximize the accuracy within a latency and power constraint. However, all of these works assume a classic single-server FL architecture. While such a one-hop architecture simplifies system design, it suffers from long-range communication costs in large systems~\cite{a22,a23}. 

% As an alternative, fully decentralized FL frameworks have been explored to eliminate the need for a central server. In decentralized FL, clients exchange model updates directly with each other via D2D links, and each client performs local model aggregation without a central coordinator \cite{a24,a25,a26,a27,a28,a29,a30}. This approach removes the single point of failure and potentially distributes the communication load. However, its performance is highly sensitive to the network topology. In well-connected peer-to-peer networks, models can propagate quickly but at the cost of heavy communication overhead, whereas in sparsely connected networks, communication is limited but the global model may fail to converge—especially when data distributions are non-IID across clients. These trade-offs highlight the practical challenges of a fully decentralized FL.

%9.16改到这里了
%这一段再读读继续改
To reduce the communication delay, multi-server FL architectures have been proposed. A prominent example is hierarchical federated learning (HFL)~\cite{a7,a14,a31}, which introduces multiple servers to collect the local models of clients within covered cells. In HFL, ESs first aggregate updates from their local clients, and then send the aggregated result to the CS for global aggregation and broadcast. In \cite{a14}, Lim et al. design dynamic resource allocation and incentive mechanisms for HFL. The hierarchical design enables more flexible communication strategies for coordinating model exchanges among clients, ESs, and the CS, which enhances both training efficiency and system scalability. In \cite{a31}, Chen et al. develop a multi-hop in-network aggregation method that averages models across intermediate nodes on the path from ESs to the CS. \cite{a32,a33} integrate centralized and decentralized FL paradigms by allowing clients within each cluster to exchange models via device-to-device (D2D) communication, while a designated cluster head uploads the aggregated model to the central server for global aggregation. However, the reliance on frequent ES-to-cloud communication and the impact of stragglers in heterogeneous environments can still cause high latency in HFL systems. As an alternative, Clustered Federated Learning (CFL)~\cite{a36,a37,a38,a39} has emerged as another prominent multi-server FL paradigm, where clients are grouped into clusters based on data or model similarity, and each cluster trains a separate model tailored to its members. In \cite{a36}, Sattler et al. propose a model-agnostic distributed multitask FL framework that clusters clients adaptively during training to enhance model relevance under NON-IID data distributions. While CFL enhances personalization and reduces inter-cluster variance, it often requires repeated clustering steps across communication rounds to adapt to client dynamics, leading to additional communication and computation overhead.

In \cite{a17}, a new FL architecture utilizing multiple servers is studied, which exploits the realistic deployment in 5G-and-beyond networks where the coverage areas of neighboring edge servers may overlap and some clients are shared between cells. The key idea is to let OCs train thier local models based on the average of received multiple edge models. Building upon this framework, \cite{a40} further proposes a two-sided learning rates algorithm and extends the theoretical analysis to non-convex loss functions. However, these algorithms incur additional delay, as OCs must wait for multiple model broadcasts before proceeding with training. Moreover, when the number of OCs is limited, their impact on overall model performance diminishes. 

\section{System Model}
\subsection{Federated Learning Basics}
A classical single-server FL system consists of a CS and \(K\) clients indexed by the set \(\mathcal{K} = \{1, 2,  \dots, K\}\), where each client \(k \in \mathcal{K}\) holds a local dataset \(\mathcal{D}_k\) of size \(n^{(k)}\). FL aims to solve the following empirical risk minimization problem:
\begin{align}
\label{equ1}
\min_{\bm{w}}  \ell (\bm{w}) = \min_{\bm{w}} \sum_{k=1}^{K} \frac{n^{(k)}}{N} \ell_{k}({\bm{w}}),
\end{align}
where \(\ell_{k}({\bm{w}}) \triangleq \mathbb{E}_{\xi \sim \mathcal{D}_k} \left[ \ell(\bm{w}, \xi) \right]\) is the local loss function. To minimize the global loss function \(\ell (\bm{w})\), the classical FL algorithm FedAvg proceeds in iterative communication rounds~\cite{a3}. In each round, clients first download the current global model from the server, then perform several local update steps (typically via stochastic gradient descent) on their local data. After local training, clients upload their updated models to the server. The server then aggregates the received updates by computing a weighted average of the local models to form a new global model:

\begin{align} 
\label{equ2}
\boldsymbol{w} = \frac{\sum_{k=1}^{K} n^{(k)} \boldsymbol{w}^{(k)}}{N},
\end{align}
where \(\boldsymbol{w}^{(k)}\) is the locally updated model of client \(k\). 

\subsection{Network Architecture: Multiple ESs With Overlapping Areas}
A single-server FL system may suffer from significant communication delays when some clients are located far from the server, making it less suitable for large-scale network deployments. In this paper, we consider a multi-server FL network motivated by realistic 5G-and-beyond scenarios, where the increasingly dense deployment of ESs result in individual users being concurrently covered by multiple ESs with reliable communication links. The system consists of a CS and a set of \(L\) ESs indexed by \(\mathcal{L} = \{1,2,\dots, L\}\), where each ES serves a group of clients within its coverage area, and some clients fall within the overlapping regions of multiple ESs. We call the local coverage of each edge server a \textit{cell}. We classify clients into two types based on their location: local clients (LCs), which are located in non-overlapping regions and communicate with only one ES, and OCs, which reside in overlapping regions and can communicate with multiple ESs.

\begin{figure}[t]
  \centering
  \includegraphics[width=1\linewidth]{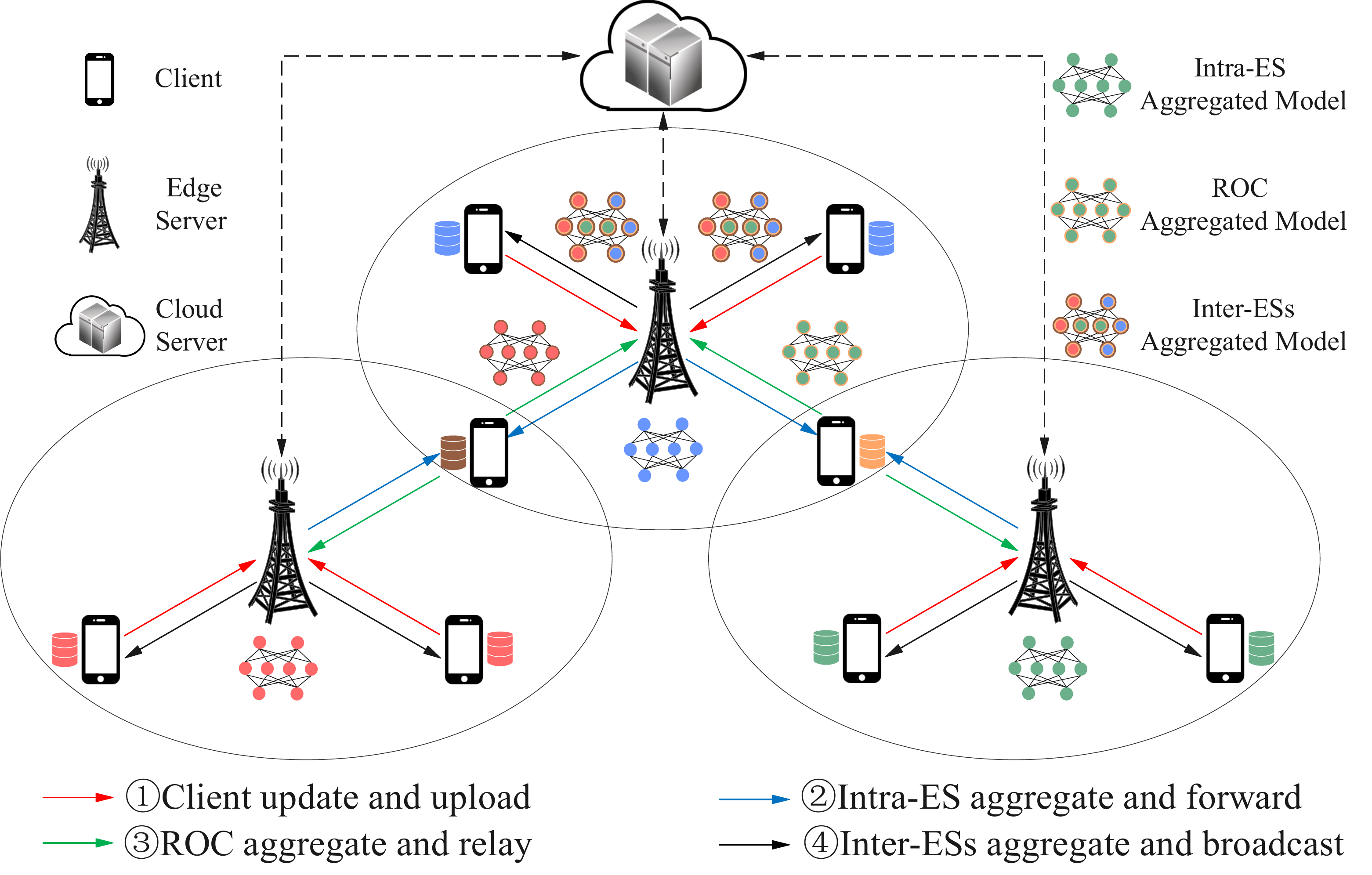}
  \caption{Illustration of a hierarchical FL network with overlapping regions.}
  \label{fig:system_model}
\end{figure}

For tractability, we assume each overlapping region involves exactly two adjacent cells and no single region is covered by more than two ESs. As illustrated in~Fig.~\ref{fig:system_model}, we adopt a chain topology to model the logical overlap between ESs, where each ES \(l\) overlaps with its two immediate neighbors ES \(l - 1\) and \(l + 1\), except for the boundary ESs, i.e., ES \(1\) and ESs \(L\), which only overlap with one adjacent ES. While not necessarily implying a physical line, this structure captures typical cases of localized overlaps, e.g., due to limited signal range or linear deployment patterns like roads or urban blocks. It simplifies analysis while preserving the key challenges posed by overlapping coverage and multi-server coordination. For generality, we do not assume overlap between the first and last ESs unless specified, since a full-loop overlap is not always present in realistic scenarios. However, our conclusions can be easily extended to loop topologies. We now formalize the client classification and related notation. Let \(\mathcal{U}_l\) denote the set of indices of LCs in cell \(l\), with cardinality \(|\mathcal{U}_l| = U_l\). Let \(\mathcal{V}_{l,l+1}\) denote the set of OCs located between cells \(l\) and \(l+1\) with the size of \(V_{l,l+1} = |\mathcal{V}_{l,l+1}|\). In each overlapping region \(\mathcal{V}_{l,l+1}\), we designate a single client as the Relay Overlapping Client (ROC), denoted by \(b_{l,l+1}\), responsible for forwarding models between the two ESs. The remaining OCs are referred to as Normal OCs (NOCs), and we define their set as \(\mathcal{N}_{l,l+1}\), such that \(\mathcal{V}_{l,l+1} = \mathcal{N}_{l,l+1} \cup \{b_{l,l+1}\}\). For readers’ convenience, the frequently used notations are summarized in Table~\ref{tab:table0}.

% In the proposed framework, OCs have two important functions: (i) acting as relay nodes that forward model information between neighboring ESs, and (ii) acting as adaptive compute nodes that can use information from multiple ESs to improve their local training. By performing these roles, OCs enable neighboring cells to share updates every round to accelerate model convergence and reduce cloud aggregation frequency.

We now formalize the client classification and related notation. Let \(\mathcal{U}_l\) denote the set of indices of LCs in cell \(l\), with cardinality \(|\mathcal{U}_l| = U_l\). Let \(\mathcal{V}_{l,l+1}\) denote the set of OCs located between cells \(l\) and \(l+1\) where \(V_{l,l+1} = |\mathcal{V}_{l,l+1}|\). In each overlapping region \(\mathcal{V}_{l,l+1}\), we designate a single client as the Relay Overlapping Client (ROC), denoted by \(b_{l,l+1}\), responsible for forwarding models between the two ESs. The remaining OCs are referred to as Normal OCs (NOCs), and we define their set as \(\mathcal{N}_{l,l+1}\), such that \(\mathcal{V}_{l,l+1} = \mathcal{N}_{l,l+1} \cup \{b_{l,l+1}\}\). For readers’ convenience, the frequently used notations are summarized in Table~\ref{tab:table0}.

\begin{table}[t]
\centering
\caption{Frequently Used Nomenclature and Notations}  % 标题不加粗，首字母正常大小，其余字体较小 
\label{tab:table0}
\begin{tabularx}{\linewidth}{l|X}
\hline
\textbf{Notation} & \textbf{Description} \\
\hline
\(n^{(k)}\) & The number of data samples of client \(k\) \\
\hline
\( \mathcal{D}^{(k)} \) & The dataset of client \(k\) \\
\hline
\( \mathcal{U}_{l} \) & The set of local clients in cell \(l\)\\
\hline
\( \mathcal{K}_{l} \) & The number of local clients in cell \(l\)\\
\hline
\( \mathcal{V}_{l,l+1} \) & The set of clients in the overlapped area between cell \(l\) and cell \(l+1\) \\
\hline
\( \mathcal{N}_{l,l+1} \) & The set of NOCs in the overlapped area between cell \(l\) and cell \(l+1\) \\
\hline
\( b_{l,l+1} \) & The relay overlapping client (ROC) in the overlapped area between cell \(l\) and \(l+1\) \\
\hline
\(R\) & The total number of rounds \\
\hline
\(\alpha_{r,l}^{(k)}\) & A binary indicator for whether client \( k \) in \( \mathcal{V}_{l,l+1} \) trains and uploads to cell \( l \) in round \( r \) \\
\hline
\(\mathcal{S}_{l}\) & The set of clients transmitting local models to ES \(l\) \\
\hline
\(\tilde{N}_{r}^{(f_l)}\) & The total number of data volume clients in \(\mathcal{S}_{l}\) \\
\hline
\(\bm{\tilde{w}}_{r,E}^{(f_l)}\) & The initial cell model of cell \( l \) calculated by averaging the models of clients in \( \mathcal{S}_l \) \\
\hline
\( \bm{w}_{r,E}^{(l+1, l)} \) & The model aggregated from cell \(l+1\) and the relay overlapping client (ROC) \(b_{l+1,l}\), which is forwarded from ES \(l+1\) to ES \(l\) \\
\hline
\( \bm{w}_{r+1}^{(f_l)} \) & The edge model of cell \( l \) for broadcasting to clients at the start of round \( r+1 \) \\
\hline
\end{tabularx}
\end{table}

\section{Proposed FedOC Algorithm}
\subsection{Algorithm Description}
Now we describe our FL algorithm tailored to the above setup. As shown in Fig.~\ref{fig:system_model}, the training process can be summarized as five stages, i.e., 1) local model update and upload. 2) intra-ES model aggregation and forwarding to ROCs. 3) ROC aggregation and relay to neighboring ESs. 4) Inter-ES aggregation and broadcasting. 5) Periodic cloud aggregation.

At the beginning of each round, clients download the edge models from ESs and update their local models, with operations varying depending on client roles: 

For each LC \(k\), local training is performed for \(E\) iterations using stochastic gradient descent (SGD) on the received edge model. The update in iteration \(e\) is computed as

\begin{align} 
\label{equ3}
\bm{w}_{r,e+1}^{(k)} = \bm{w}_{r,e}^{(k)} - \eta_{r,e} \nabla \ell_k (\bm{w}_{r,e}^{(k)}),
\end{align}
where \(e = 0, 1, \ldots, E-1\), and \( \nabla \ell_k (\bm{w}_{r,e}^{(k)})\) is the gradient of the local loss function. After local updating, each LC \(k\) transmits the local model \(\bm{w}_{r,E}^{(k)}\) to corresponding edge server. 

For OCs, an overlapping client \(k \in \mathcal{V}_{l,l+1}\) could receive two models: \(\bm{w}_r^{(f_l)}\) from ES \(l\) and \(\bm{w}_r^{(f_{l+1})}\) from ES \(l+1\). Considering the computation and communication heterogeneity of clients across different ESs, the model broadcast times of ESs may differ significantly. Therefore, requiring each overlapping client (OC) to wait for all edge models before starting local training may not meet real-time requirements. To handle this, we propose a \textit{Fastest Selection Strategy}, in which each OC dynamically adopts the earliest-arriving model from the available ESs as the initialization for local training. The OC’s initial training model in round \(r\) can be written as

\begin{align}
\label{equ4}
\bm{w}_{r}^{(k)} &= \alpha_{r,l}^{(k)}\bm{w}_r^{(f_{l})} + \alpha_{r,l+1}^{(k)}\bm{w}_r^{(f_{l+1})}, 
\end{align}
where \(\alpha_{r,l}^{(k)}, \alpha_{r,l+1}^{(k)} \in \{0, 1\}\) indicate whether the initial training model \(\bm{w}_{r}^{(k)}\) is selected from ES \(l\) or ES \(l+1\), and we have \(\quad \alpha_{r,l}^{(k)} + \alpha_{r,l+1}^{(k)} = 1\). Then each OC \(k\) executes SGD updates as LCs. Finally, each NOC transmits \(\bm{w}_{r,E}^{(k)}\) to ES \(i\) when \(\alpha_{r,i}^{(k)} = 1\) holds, ROCs keep their updated models locally and upload them during the subsequent phase. 
This is intentional: the ROC could forward its own local model alongside the relayed ES model during the next stage, thereby reducing redundant transmissions.

In the ES local aggregation and forwarding stage, each ES \(l\) receives local models from clients set \(\mathcal{S}_{r}^{(l)}\) located at its cell, where \(\mathcal{S}_r^{(l)}\) consists of LCs and the NOCs that upload local models, which can be written as \(
\mathcal{S}_r^{(l)} = \mathcal{U}_{l} \cup \{k \mid \alpha_{r,l}^{k} = 1, \; k \in (\mathcal{N}_{l-1,l} \cup \mathcal{N}_{l,l+1})\}\). First, each ES \(l\) aggregates the received local models to obtain the initial cell model \(\bm{\tilde{w}}_{r}^{(f_l)}\) according to the following equation:

\begin{align} 
\label{equ5}
\bm{\tilde{w}}_{r,E}^{(f_l)} = \frac{\sum_{k \in \mathcal{S}_r^{(l)}} n^{(k)} \boldsymbol{w}_{r,E}^{(k)}}{\tilde{N}_{r}^{(f_l)}},
\end{align}
where \(\tilde{N}_{r}^{(f_l)} = \sum_{k \in \mathcal{S}_r^{(l)}} n^{(k)}\). Then, each ES \(l\) will transmit its aggregated model to neighboring ESs through ROCs.

In the ROC aggregation and relay stage, the exchange between neighboring edge servers occurs through the ROCs, and the relay rule is: once ROC \( b_{l,l+1} \) receives an edge model from either ES \( l \) or ES \( l+1 \), it will first aggregate this model with its own local model \(\bm{w}_{r,E}^{(b_{l,l+1})}\) trained during the model update stage, and then forward it to the other ES. Thus, the model that ROC \( b_{l, l+1} \) transmits to ES \( l \) can be written as

\begin{align} 
\label{equ6}
\bm{w}_{r,E}^{(l+1, l)} = \frac{\tilde{N}_{r}^{(f_{l+1})} \bm{\tilde{w}}_{r,E}^{(f_{l+1})} + n^{(b_{l,l+1})} \bm{w}_{r,E}^{(b_{l,l+1})}}{N_{r}^{(l+1,l)}},
\end{align}
where \(N_{r}^{(l+1,l)} = \tilde{N}_{r}^{(f_{l+1})} + n^{(b_{l,l+1})}\). 

For each ES \(l\), after receiving the models and corresponding data volume from ROCs \(b_{l,l+1} \text{ and } b_{l-1,l}\), ES \(l\) will aggregate these models with its cell models to update edge model \(\bm{w}_{r+1}^{(f_l)}\) for broadcast in \(r+1\)-th iteration, which can be written as 

\begin{align} 
\label{equ7}
\bm {w}_{r+1}^{(f_l)}
= \frac{N_{r}^{(l+1,l)} \bm{w}_{r,E}^{(l+1, l)} + N_{r}^{(l-1,l)} \bm{w}_{r,E}^{(l-1, l)} + \tilde{N}_{r}^{(f_l)} \bm{\tilde{w}}_{r,E}^{(f_l)}}{N_{r}^{(l+1,l)} + N_{r}^{(l-1,l)} + \tilde{N}_{r}^{(f_l)}}.
\end{align}
For boundary cells, the missing neighbor’s model is taken as zero. so \eqref{equ7} holds for all \(l\in\{1,\ldots,L\}\).

After every \(\kappa\) rounds, all clients transmit their local models to the corresponding ESs, which then upload the aggregated cell models to the CS for further aggregation and broadcasting back to the ESs. These processes will repeat \(R\) rounds.

\subsection{Latency Analysis}
%因为我们考虑的是cell之间存在overlap的场景，所以不同cell之间的通信干扰问题是必须考虑的。假设我们总共有B总频率带宽用于模型传输。为了避免干扰，相邻cell之使用不同的频段，所以每个cell分配R/2带宽。每个ES将为cell内需要上传的客户端分配带宽用于上传本地模型，具体分配策略超出了我们这篇论文讨论的范畴，我们假设平均分配R/2带宽。而后等到所有客户端模型上传后，释放R/2带宽，edge利用R/2的带宽传递cell内聚合的模型通过ROCs进而传递给邻居的ESs。若要传递给两边ROC则每条ES-ROC-ES为R/4带宽，对于边缘的ES只需要传递给一个邻居ES，故而可以直接利用R/2带宽。因为不同cell之间的带宽不干扰，每个cell在完成自己内部的模型传输和聚合后就可以立即传递给邻居cell。和HFL相比，我们只是增加了ES通过ROC中继传递给邻居edge的通信，由于原本用来全cell内所有客户端的带宽用来传输1-2条通信链路，所以这个通信延时几乎可以忽略。 edge和云端之间的通信与传统的HFL是完全类似的，在此不做讨论。

Since neighboring cells overlap, inter-cell interference cannot be ignored. Assume that a total spectrum budget of $B$ is reserved for model transmission. To eliminate interference, adjacent cells operate on disjoint sub-bands, so each cell is allocated $B/2$. Within a cell, each ES \(l\) first broadcasts the edge model to clients and then allocates its $B/2$ bandwidth among the clients for uploading local models. The details of this intra-cell scheduling are beyond the scope of this study, and we simply assume equal allocation among clients. We denote \(t_{\text{cast}}^{(l)},~t_{\text{comp}}^{(l)},~t_{\text{upload}}^{(l)}\) 
as the times for edge broadcasting, maximal local model training, and maximal local model uploading in cell \(l\), respectively. And we have
\begin{align}
\label{equ8}
t_{\text{upload}}^{(l)} = \frac{M}{\frac{B}{2|S_{l}|} \log (1 + \frac{\min_{k \in \mathcal{S}_l} \{ \delta_{k} \} p}{B/(2|S_{l}|) \cdot N_0})}
\end{align}
where \(M\) is the model size, \(\delta_{l,k}\) is the channel gain of client \(k \in \mathcal{S}_l\) to ES \(l\), \(p\) is the transmit power of client, and \(N_0\) is the noise spectral density.

We denote \(t_{\text{edge}}^{(l)} = t_{\text{cast}}^{(l)}+t_{\text{comp}}^{(l)}+t_{\text{upload}}^{(l)}\) as the inner edge latency of ES \(l\). After aggregating the local uploading models, the ES reclaims the same $B/2$ bandwidth and forwards the aggregated model to its neighboring ESs through relay nodes ROCs. If the ES \(l\) overlaps with two neighbors, every ROC relay link will using \(B/4\) bandwidth, the communication latency can be written as
\begin{align}
\label{equ9}
t_{\text{relay}}^{(l,l+1)}
&= \frac{M}{\tfrac{B}{4}\left[
  \log\!\left(1+\tfrac{4\delta_{l,l+1}P}{BN_0}\right)
 +\log\!\left(1+\tfrac{4\delta_{l,l+1}p}{BN_0}\right)
 \right]} .
\end{align}

where \(\delta_{l,l+1}\) is the channel gain of the link from ES \(l\) to ES \(l+1\) via ROC \(b_{l-1,l}\), and \(P\) is the ES' transmit power. The relay latency of ES \(l\) should be \(t_{\text{relay}}^{(l)} = \max\{t_{\text{relay}}^{(l,l+1)},t_{\text{relay}}^{(l,l-1)}\}\), where \(t_{\text{relay}}^{(1,0)} = t_{\text{relay}}^{(L,L+1)}=0\). Compared with conventional HFL, the only additional step is this brief ES-ROC-ES hops. Compared with \eqref{equ8} and \eqref{equ9}, since usually \(|S^{(l)} \gg  2|\), a bandwidth large enough for all client uploads now carries at most two links, the resulting delay is negligible. After every \(\kappa\) rounds, the cloud aggregation happens, and we assume the maximal latency from all clients to the CS as \(t_{\text{cloud}}\).

\section{THEORETICAL RESULTS}
We consider three cells, and a \(C\) class classification problem defined over a compact space \(\mathcal{X}\) and a label space \(\mathcal{Y} = \{1, 2, \ldots, C\}\). Data point \((\bm{x},y)\) is distributed across \(\mathcal{X} \times \mathcal{Y}\) following the probability distribution \(P\). Inspired by \cite{a34}, we define the population loss \(\ell (\bm{w})\) with the widely used cross-entropy loss as \(\ell(\boldsymbol{w})=\mathbb{E}_{\boldsymbol{x}, y \sim P}\left[\sum_{i=1}^C \mathbbm{1}_{y=i} \log f_i(\boldsymbol{x}, \boldsymbol{w})\right]=\sum_{i=1}^C P_{y=i} \mathbb{E}_{\boldsymbol{x} \mid y=i}\left[\log f_i(\boldsymbol{x}, \boldsymbol{w})\right]\), where \(P_{y=i}\) represents the proportion of data belonging to class \(i\) relative to the total data set. Thus the local update in \eqref{equ3} can be written as

\begin{align} 
\label{equ10}
\bm{w}_{r,e+1}^{(k)} = \bm{w}_{r,e}^{(k)} - \eta_{r,e} \sum_{i=1}^C P_{y=i}^{(k)} \nabla_{\bm{w}} \mathbb{E}_{\boldsymbol{x} \mid y=i}\left[\log f_i(\boldsymbol{x}, \bm{w}_{r,e}^{(k)})\right],
\end{align}
where \(P_{y=i}^{(k)}\) represents the data proportion of class \(i\) on clients \(k\), and \(f_i\) is a function denoting the probability for the \(i\)-th class, which is parameterized over the weight of the neural network \cite{a34}. For simplicity, we assume OCs participate in training via a fixed assignment strategy that does not change across rounds in FedOC, and we mainly explore the convergence performance of the neighboring edge aggregation strategy.

% Another one is the global centralized SGD algorithm, where the data in all clients is collected in the CS to train a global model by using SGD algorithm. We use \(w^{*}\) as the final model of global centralized SGD. 

Now we introduce three widely used assumptions in FL:

\begin{assumption}
\label{assumption1} \(\ell(\boldsymbol{w})\) is \(\lambda\)-smooth, i.e., for all \(\bm{v}\) and \(\bm{w}\), \(\|\nabla \ell(\bm{w})-\nabla \ell(\bm{v})\| \leqslant \lambda\|\bm{w}-\bm{v}\|\).
\end{assumption}

\begin{assumption}
\label{assumption2} \(\nabla_{\boldsymbol{w}} \mathbb{E}_{\boldsymbol{x} \mid y=i} \log f_i(\boldsymbol{x}, \boldsymbol{w})\) is \(\lambda_{\boldsymbol{x} \mid y=i}\)-Lipschitz for each class \(i \in [C]\), i.e., for all \(\bm{v}\) and \(\bm{w}\), \(\|\nabla_{\boldsymbol{w}} \mathbb{E}_{\boldsymbol{x} \mid y=i} \log f_i(\boldsymbol{x}, \boldsymbol{w}) - \nabla_{\boldsymbol{w}} \mathbb{E}_{\boldsymbol{x} \mid y=i} \log f_i(\boldsymbol{x}, \boldsymbol{v}) \| \leq \lambda_{\boldsymbol{x} \mid y=i}\|\bm{w}-\bm{v}\|\).
\end{assumption}

\begin{assumption}
\label{assumption3} There exists a constant \(\Delta_{\max} > 0\) such that \(\left\|\nabla \ell_k(\boldsymbol{w}^{(k)}_{r,e})\right\| \leq {\Delta}_{\text{max}}, \forall k\).
\end{assumption}

To derive the convergence bound of the loss function, we introduce a cell-centralized SGD algorithm where the clients' datasets are centralized in respective ESs for model training and uploading to the CS for aggregation. Based on the above assumptions, we derive the model divergence between FedOC and cell-centralized SGD as follows:

\begin{lemma}
\label{lemma}
Suppose Assumptions 1-2 hold, then we have the following inequality:
\begin{align}
\label{equ11}
&\bigg\|\bm{w}_{R}^{(f)} - \bm{w}_{R}^{(c)}\bigg\| \notag\\
\leq&\underbrace{\sum_{p=R-\kappa}^{R-1}\left(\prod_{q=p+1}^{R-1} D_q \right)\epsilon_{p}^{\text{intra}}}_{\epsilon^{\text{intra}}} + \underbrace{\sum_{p=R-\kappa+1}^{R-1}\left(\prod_{q=p+1}^{R-1} D_q\right)\epsilon_{p}^{\text{inter}}}_{\epsilon^{\text{inter}}},
\end{align}
where \(\epsilon_{R-1}^{\text{intra}} = \sum\limits_{j=1}^{3}  \frac{\sum_{k \in \hat{\mathcal{K}}^{(f_j)}} 
   n^{(k)}}{N}
   \left(\sum_{i=1}^C \bigl|P_{y=i}^{(k)} - P_{y=i}^{(c_j)}\bigr|\right) \cdot \\ 
   \sum_{e=0}^{E-2} \eta_{R-1,e} \left(\prod_{d=e+1}^{E-1} a_{R-1,d}^{(k)}\right)
   g_{\max}\!\bigl(\boldsymbol{w}_{R-1,e}^{(c_j)}\bigr)\), \(\epsilon_{R-1}^{\text{inter}} = D_{R-1}^{(1)} \bigg\| \sum_{j=1}^{3}\rho_{R-2}^{(j)} \boldsymbol{\hat{w}}_{R-2, E}^{(f_j)}\bigg\|+D_{R-1}^{(3)} \bigg\| \sum_{j=1}^{3} \mu_{R-2}^{(j)}\boldsymbol{\hat{w}}_{R-2, E}^{(f_j)}\bigg\|\) and \(a_{R-1,d}^{(k)},~g_{\max}\!\bigl(\boldsymbol{w}_{R-1,e}^{(c_j)}\bigr),~D_{R-1}^{(j)},~\rho_{R-2}^{(j)},~\mu_{R-2}^{(j)}\) are constants denoted in Appendix A.
\end{lemma}
\begin{proof}
See Appendix A
\end{proof}

This lemma indicates that the model divergence between FedOC and cell-centralized SGD mainly stems from two sources. The first term, \(\epsilon_{R-1}^{\text{intra}}\), measures the error of intra-cell data heterogeneity, which is caused by the data distribution's divergence between each client and the whole cell, i.e., \(\sum_{i=1}^{C}\left| P_{y=i}^{(k)} - P_{y=i}^{(c_j)} \right|\). The second term, \(\epsilon_{R-1}^{\text{inter}}\), reflects the model divergence induced by the neighboring ES model aggregation strategy. Since \(\sum_{j=1}^{3}\rho_{R-2}^{(j)}=\sum_{j=1}^{3} \mu_{R-2}^{(j)}=0\), \(\epsilon_{R-1}^{\text{inter}}\) vanishes when the edge models coincide across cells.

\begin{theorem}
\label{theorem}
Suppose Assumptions 1-3 hold, and set \(\bm{w}^{*}\) is the optimal global model. If the learning rate satisfies \(\eta_{r,e} = \frac{1}{r(E -1)} \quad \text{for }~ 0 \leq r \leq R~\text{and}~ 0\leq e \leq E-1\), then we have

\begin{align} 
\label{equ12}
& \ell(\bm{w}_{R+1}^{(f)}) - \ell(\bm{w}^{*}) \notag\\
\leq& \frac{\lambda}{2} \Biggl[\frac{\sum_{j=1}^{3} \hat{N}^{(f_j)} H^{(j)}\left( \sum_{i=1}^{C} \left\| P_{y=i}^{(c_j)} - P_{y=i}^{(c)} \right\|\right) }{N(R-1)(E-1)} + \notag\\
&\underbrace{\frac{\kappa \beta_{\text{max}} \prod_{p =R-\kappa+1}^{R-1} D_p}{R-\kappa}}_{\epsilon^{\text{intra}}} + \underbrace{\frac{\kappa(\kappa-1) D_{\text{max}} \bar{\Delta}_{\text{max}} \prod_{i=R-\kappa+2}^{R-1} D_{i}}{2(R-\kappa)(E-1)}}_{\epsilon^{\text{inter}}}\Biggr],
\end{align}
where \(H^{(j)}, \beta_{\text{max}},~D_{\text{max}},~ \bar{\Delta}_{\text{max}}~and~D_p,~for~p \leq R-1\)~are~constants~denoted~in~Appendix~B.
\end{theorem}
\begin{proof}
See Appendix B
\end{proof}

Based on \textbf{Theorem 1}, the loss error of FedOC consists of three terms. 
The first arises from the divergence between the data distributions of individual cells and the global dataset, i.e., \(\sum_{j=1}^{3} \hat{N}^{(f_j)} H^{(j)}\left( \sum_{i=1}^{C} \left\| P_{y=i}^{(c_j)} - P_{y=i}^{(c)} \right\|\right)\). The second term corresponds to intra-cell data heterogeneity, denoted by \(\epsilon^{\text{intra}}\), while the third term reflects the model divergence across different ESs, denoted by \(\epsilon^{\text{inter}}\). Moreover, the cloud aggregation interval \(\kappa\) influences the convergence rate: a larger \(\kappa\) amplifies the mismatch among ES models and therefore slows down convergence. Nevertheless, since \(\kappa\) is fixed, as \(R \to \infty\) all error terms converge to zero, which guarantees that FedOC converges to the optimal solution.
\begin{table}[t]
\caption{Simulation Parameters}
\vspace{-0.5cm}
\begin{center}
\small
\begin{tabular}{p{5.4cm}|p{2.7cm}}  % Adjust column width here
\hline
\textbf{Parameters} &    \textbf{Values} \\
\hline
Number of clients, \(K\) & \(60\)\\
\hline
Number of iterations, \(R\) & \(500\)\\
\hline
Number of local training epochs, \({U}_{k}(r)\) & \(5\)\\
\hline
Client's transmit power, \(p~\mathrm{(W)}\) & \(1\) \\
\hline
ES's transmit power, \(P~\mathrm{(W)}\) & \(5\)\\
\hline
Total channel bandwidth, \(B~\mathrm{(MHz)}\)  & \(50\)\\
\hline
Client's one-epoch update time (seconds) & MNIST: \([0.1, 0.2]\), CIFAR-10: \([1, 2]\) \\
\hline
Noise power spectral density, \({N}_{0}~\mathrm{(dBm/Hz)}\) & \(-174\)\\
\hline
The number of model parameters & for MNIST: \(21840\), for CIFAR-10: \(1.14\) million \\
\hline
Batch size & \(20\)\\
\hline
Initial learning rate, \(\eta\) & for MNIST: \(0.01\), for CIFAR-10: \(0.1\)\\
\hline 
Exponential decay factor of learning rate in SGD & for MNIST: \(0.995\), for CIFAR-10: \(0.992\)\\
\hline
\end{tabular}
\label{tab1}
\end{center}
\vspace{-0.5cm}
\end{table}
\section{Performance Evaluation}
% In this section, we evaluate the effectiveness and efficiency of FedOC through extensive simulations.

\subsection{Simulation Setup}
We consider a multi-server network consisting of \(L=3\) ESs and \(K = 60\) clients. Each ES is located at the center of a circular coverage area with a radius of \(600\) meters. For a fair comparison, multiple clients are randomly distributed across cells, where the numbers of clients in the cells satisfy \(U_1+V_{1,2}/2 = V_{1,2}/2+U_2+V_{2,3}/2= V_{2,3}+U_3\). To better reflect practical scenarios, we examine two degrees of overlapping coverage between cells: i) a minimal overlap case with only \(1\) overlapping client located in each overlapping region, and ii) a moderate overlap case with \(K/(2L)=10\) OCs in each overlap region, i.e., \(|V_{1,2}| = |V_{2,3}| = 10\) out of \(60\). The channel gain of communication link is composed of both small-scale fading and large-scale fading. The small-scale fading is set as a Rayleigh distribution with uniform variance, and the large scale fading is generated using the path-loss model of \(128.1+37.6\text{log}_{10}(d(\text{km}))\) \cite{a18}, \cite{a41}. The noise power and device power are assume to be \(-174 \text{dBm/Hz}\) and \(p=1\text{W}\), respectively. The transmit power of each ES is set to \(5\text{W}\) and the total bandwidth is \(50\text{MHz}\). We set the latency of each client’s training time per epoch as 0.1–0.2s for MNIST and 1–2s for CIFAR-10 \cite{a31}. The key simulation parameters are summarized in Table~\ref{tab1}.
%此处插入折线图
\begin{figure}[t]
    \centering
    \includegraphics[width=1\linewidth]{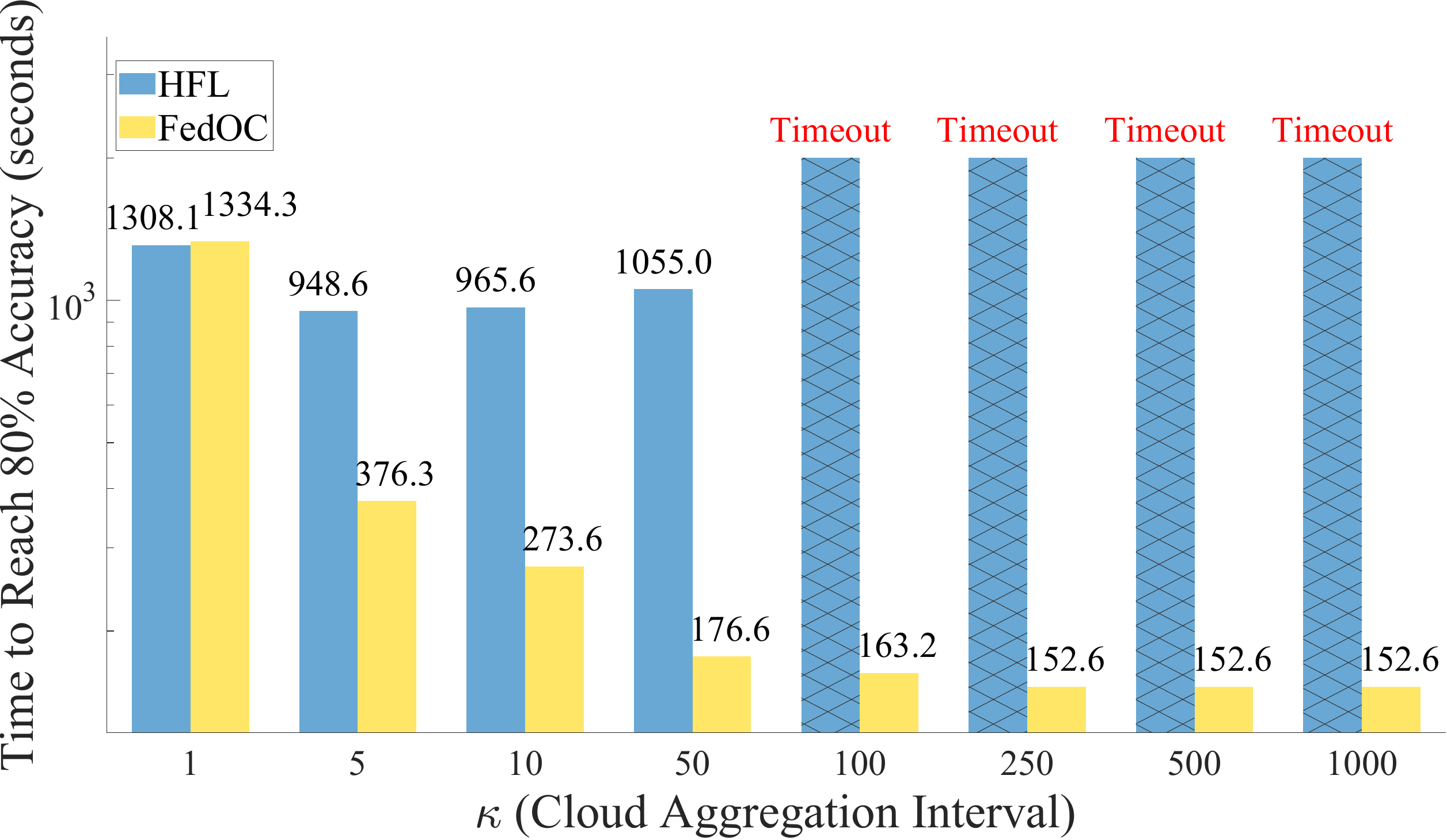}
    \caption{Time to reach target accuracy versus \(\kappa\) on MNIST. "Timeout": failed to reach target accuracy within 1600 seconds.}
    \label{fig:bars1}
\end{figure}

\begin{figure}[t]
    \centering
    \includegraphics[width=1\linewidth]{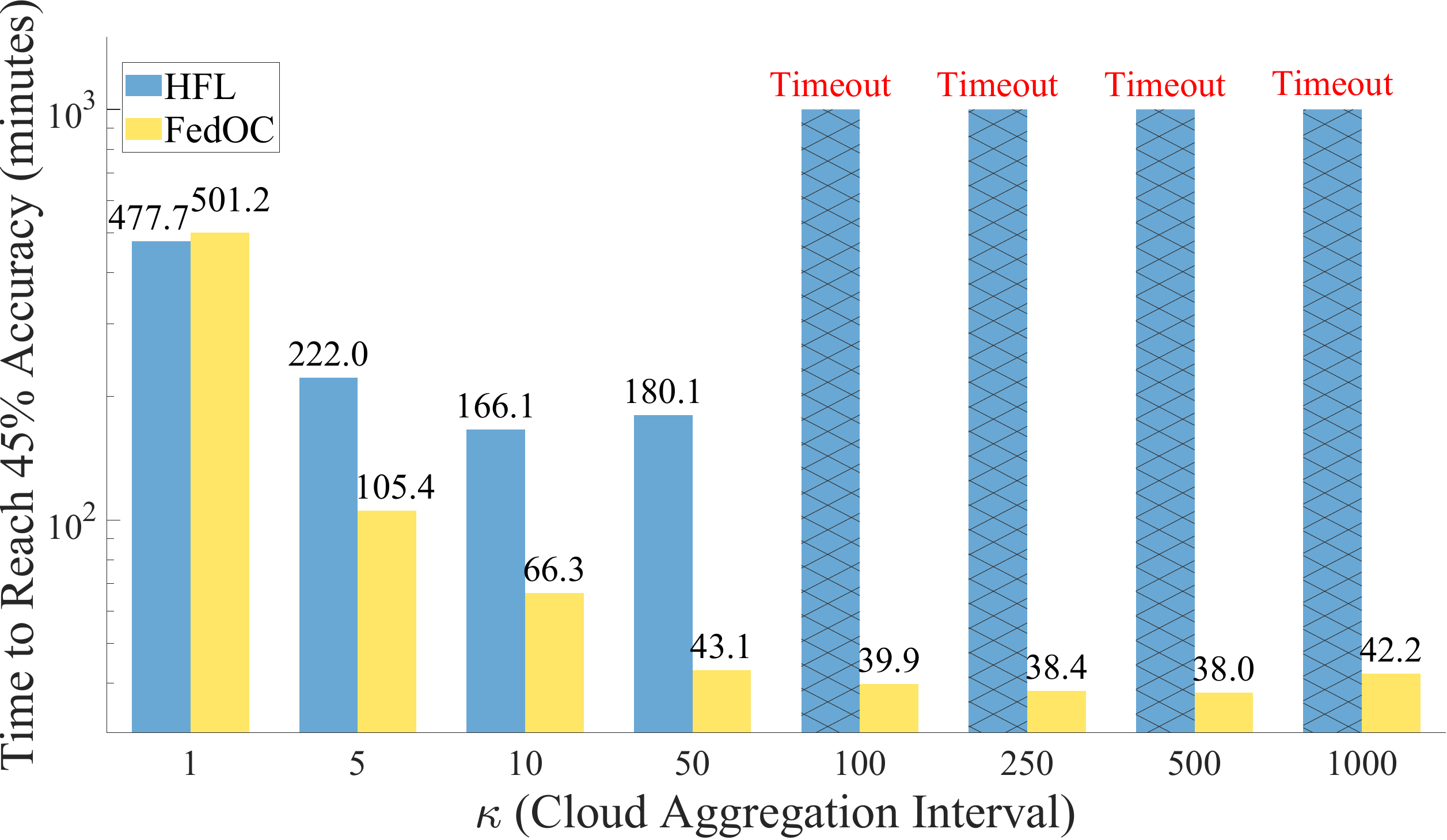}
    \caption{Time to reach target accuracy versus \(\kappa\) on CIFAR-10. "Timeout": failed to reach target accuracy within 600 minutes.}
    \label{fig:bars2}
\end{figure}

\begin{figure*}[t]
    \centering
    \subfloat[MNIST\label{fig:NON-IID-1-mnist-new}]{
        \includegraphics[width=0.48\textwidth]{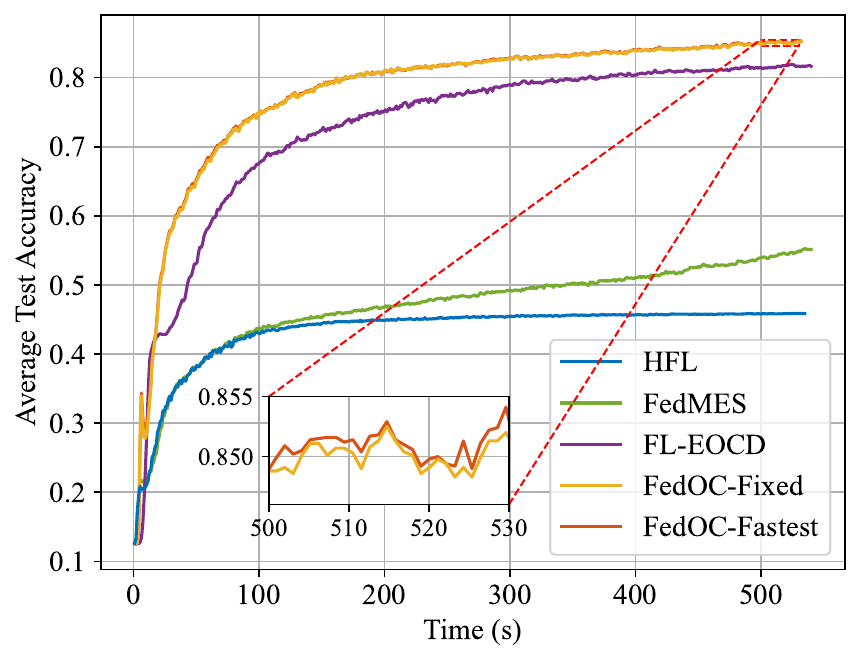}
    }
    \hfill
    \subfloat[CIFAR-10\label{fig:NON-IID-1-cifar10-new}]{
        \includegraphics[width=0.48\textwidth]{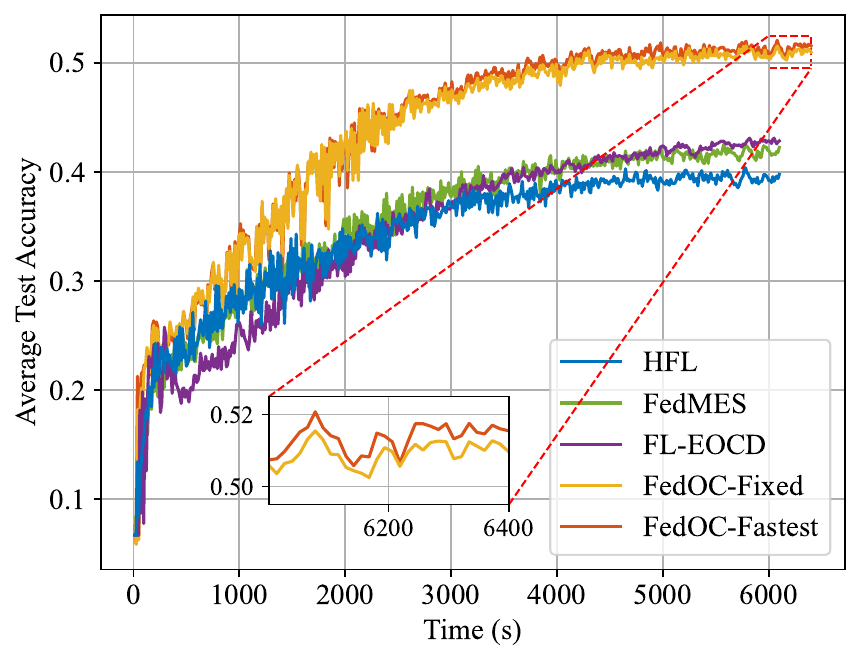}
    }
    \caption{Average test accuracy versus training time under a 3-ES deployment with \(|\mathcal{V}_{1,2}| = |\mathcal{V}_{2,3}| = 1\) setup.}
    \label{fig:NON-IID-1-overlap-new}
\end{figure*}

\begin{figure*}[t]
    \centering
    \subfloat[MNIST\label{fig:NON-IID-10-mnist}]{
        \includegraphics[width=0.48\textwidth]{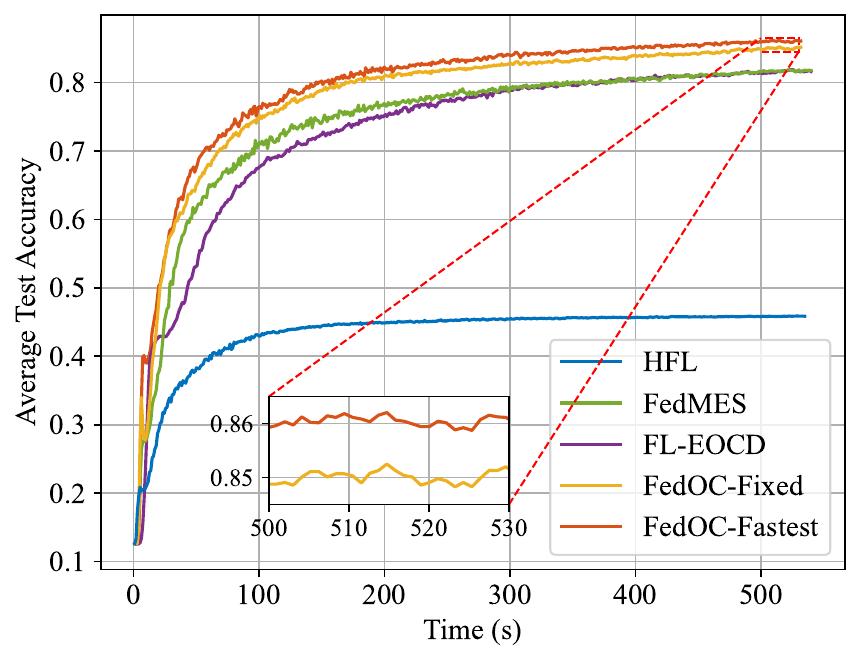}
    }
    \hfill
    \subfloat[CIFAR-10\label{fig:NON-IID-10-cifar10}]{
        \includegraphics[width=0.48\textwidth]{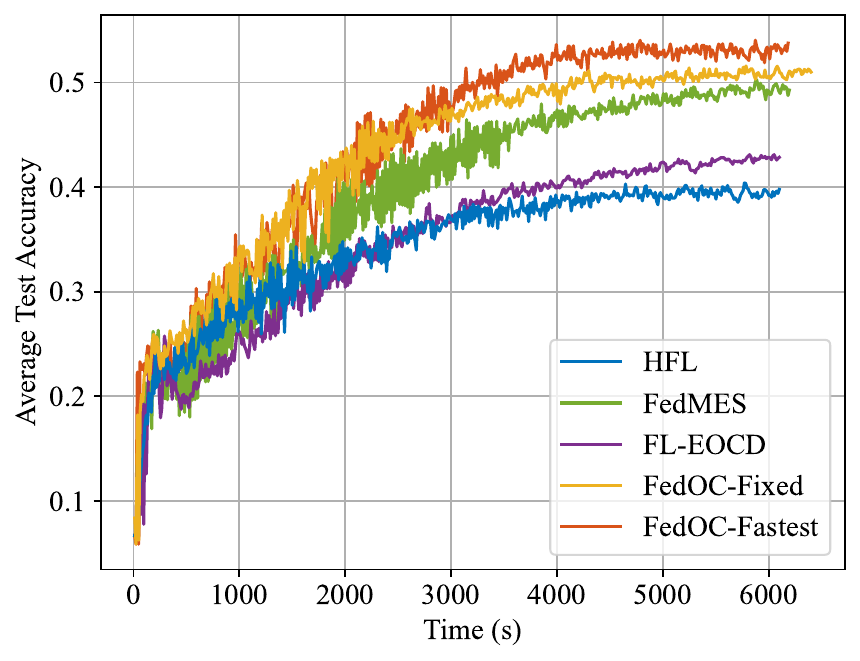}
    }
    \caption{Average test accuracy versus training time under a 3-ES deployment with \(|\mathcal{V}_{1,2}| = |\mathcal{V}_{2,3}| = 10\) setup.}
    \label{fig:NON-IID-10-overlap}
\end{figure*}
We conduct experiments using two widely adopted image classification datasets, MNIST \cite{a42} and CIFAR-10 \cite{a43}. For the MNIST dataset, we adopt a lightweight Convolutional Neural Network (CNN)  with 21,840 parameters, following the design proposed in \cite{a7}. For CIFAR-10, we employ a deeper six-layer CNN with 1.14 million parameters, as described in \cite{a31}. Unless otherwise noted, each client performs \(5\) local epochs per round before uploading its model update. We consider the data distribution heterogeneity as Client Non-IID and Cell Non-IID, where each client again holds data from 2 classes, and each edge server's cell is constrained to at most 5 unique classes, resulting in significant class imbalance across cells. 

We compare FedOC against three baseline schemes: 
\begin{itemize}
    \item \textbf{HFL\cite{a7} :} A conventional three-tier hierarchical FL that ignores overlaps, where each client is associated with a single fixed ES and the CS aggregates the ES models in each round.
    \item \textbf{FedMES\cite{a17}:} An overlapping-client method where each overlapping device averages the models from all its connected ESs after the downloading stage, trains on this merged model, and then uploads the result to all those ESs.
    \item \textbf{FL-EOCD}~\cite{a18}: A two-step variant of FedMES. The overlapping client first caches the received edge models and trains locally as in FedMES, and then aggregates the local model with the cached edge models before uploading the result to the corresponding ESs.
    \item \textbf{FedOC with fixed NOC participating strategy (FedOC-Fixed):} To demonstrate the advantage of the OC \textit{Fastest Selection strategy}, we let the clients in the overlapped regions communicate with only one fixed ES among different rounds.
\end{itemize}

For ease of distinction, we refer to our proposed FedOC algorithm with OCs under the \textit{fastest selection strategy} as FedOC-fastest in the following figures.
% We track the test accuracy over training time measured in simulated wall-clock time as our primary metric, which reflects both convergence speed and final model quality.

\subsection{Impact of \(\kappa\) on FedOC}
%从Theorem 1和Corollary 1揭示了损失误差界是受各个cell之间的inter-cell模型差距影响的，云聚合频率在其中充当重要的角色，除此之外，我们还需要综合考虑算法的模型传输效率。为此，We set t_{edge}/t_{comp} =10 to capture the communication bottleneck as in [13] and normalize t_{edge} to 1. We define t_c = t_{cloud}/t_{edge}, where tc 1 holds due to the significant communication delay between the client and the CS. We set tc =10for experiments as in the setup in [23]. Mini-batch stochastic gradient descent with momentum is utilized for training, with an initial learning rate of 0.01 and a momentum term of 0.9. For all experiments, we set the number of local epochs at each client to 5, and the local mini-batch size to 10.

As shown in \textbf{Theorem 1}, frequent global synchronization, i.e., a small \(\kappa\), can accelerate learning by rapidly reducing inter-cell model divergence, but it incurs higher latency and networking overhead each round. A key question is how often cloud aggregation should be performed to balance convergence speed against communication costs. Similarly to \(t^{(l)}_{\text{edge}}\), we define \( t^{(l)}_{\text{cloud}}\) as the time for a one-round communication between ES \(l\) and the CS. We define \( t_c = t^{(l)}_{\text{cloud}} / t^{(l)}_{\text{edge}} \), where usually \( t_c \gg 1 \) holds due to the significant communication delay between the client and the CS \cite{a17}. We set \(t_c = 10\) for experiments as in the setup in \cite{a17}. This subsection investigates how the cloud aggregation interval affects training efficiency and contrasts the performance of FedOC against that of the classical HFL approach.

% % ===== 第一组：bar 图 =====
% \begin{figure*}[t]
%     \centering
%     % 左图
%     \begin{minipage}[t]{0.48\textwidth}
%         \centering
%         \includegraphics[width=\linewidth]{bar_mnist_final.pdf}
%         \caption{Time to reach target accuracy versus $\kappa$ on MNIST. "Timeout": failed to reach target accuracy within 1600 seconds.}
%         \label{fig:bars1}
%     \end{minipage}
%     \hfill
%     % 右图
%     \begin{minipage}[t]{0.48\textwidth}
%         \centering
%         \includegraphics[width=\linewidth]{bar_cifar_final.pdf}
%         \caption{Time to reach target accuracy versus $\kappa$ on CIFAR-10. "Timeout": failed to reach target accuracy within 600 minutes.}
%         \label{fig:bars2}
%     \end{minipage}
% \end{figure*}
\subsubsection{Experiments under 3 overlapping cells}
\begin{figure*}[t]
    \centering
    \subfloat[MNIST\label{fig:NON-IID-1-mnist-4cells}]{
        \includegraphics[width=0.48\textwidth]{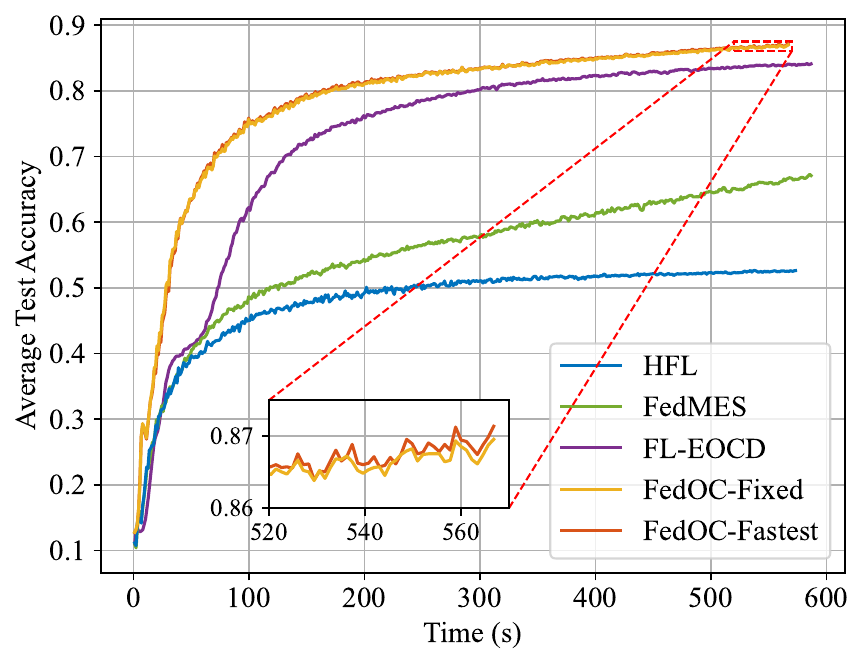}
    }
    \hfill
    \subfloat[CIFAR-10\label{fig:NON-IID-1-cifar10-4cells}]{
        \includegraphics[width=0.48\textwidth]{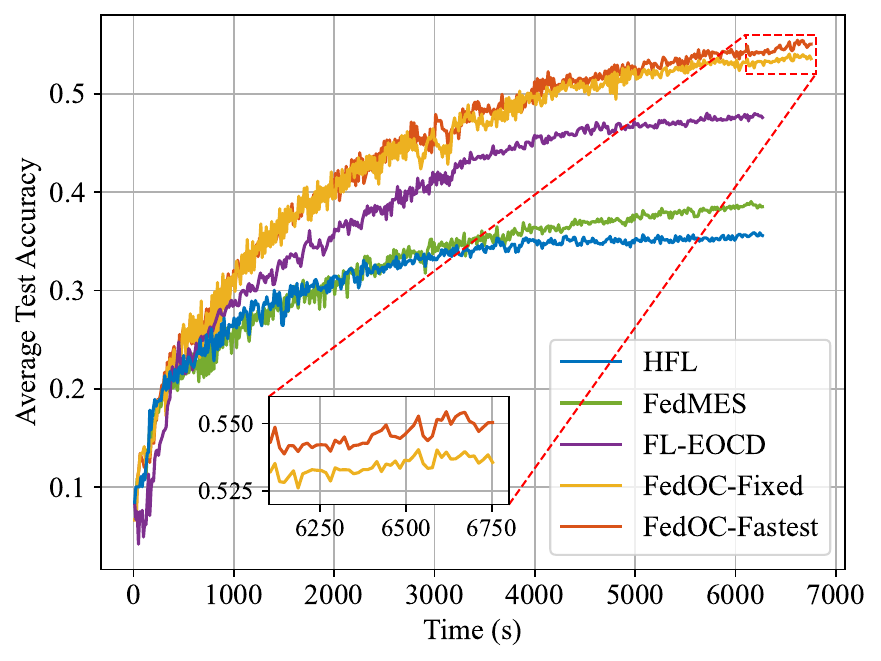}
    }
    \caption{Average test accuracy versus training time under a 4-ES deployment with \(|\mathcal{V}_{1,2}| = |\mathcal{V}_{2,3}| = 1\) setup.}
    \label{fig:NON-IID-1-overlap-4cells}
\end{figure*}

\begin{figure*}[t]
    \centering
    \subfloat[MNIST\label{fig:NON-IID-7-mnist-4cells}]{
        \includegraphics[width=0.48\textwidth]{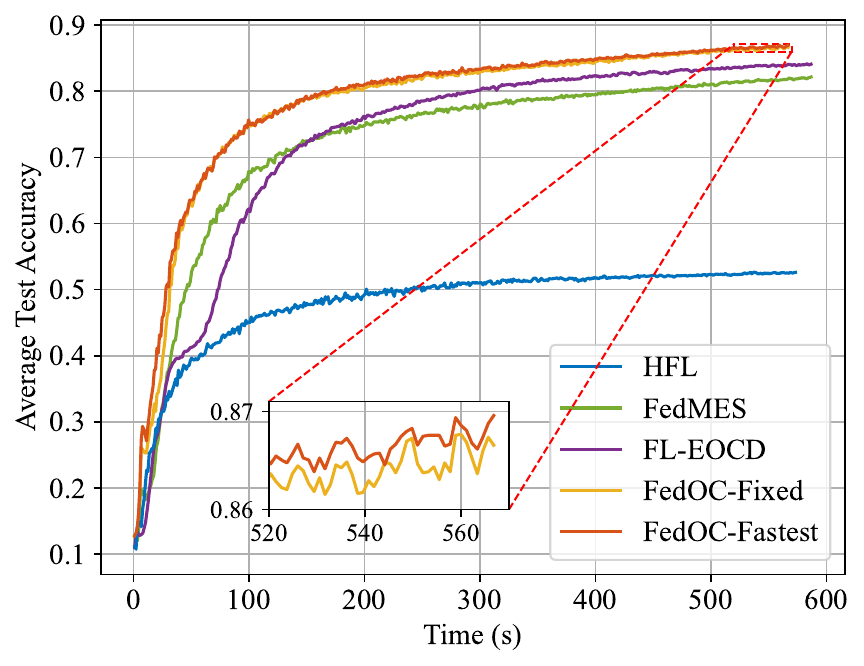}
    }
    \hfill
    \subfloat[CIFAR-10\label{fig:NON-IID-7-cifar10-4cells}]{
        \includegraphics[width=0.48\textwidth]{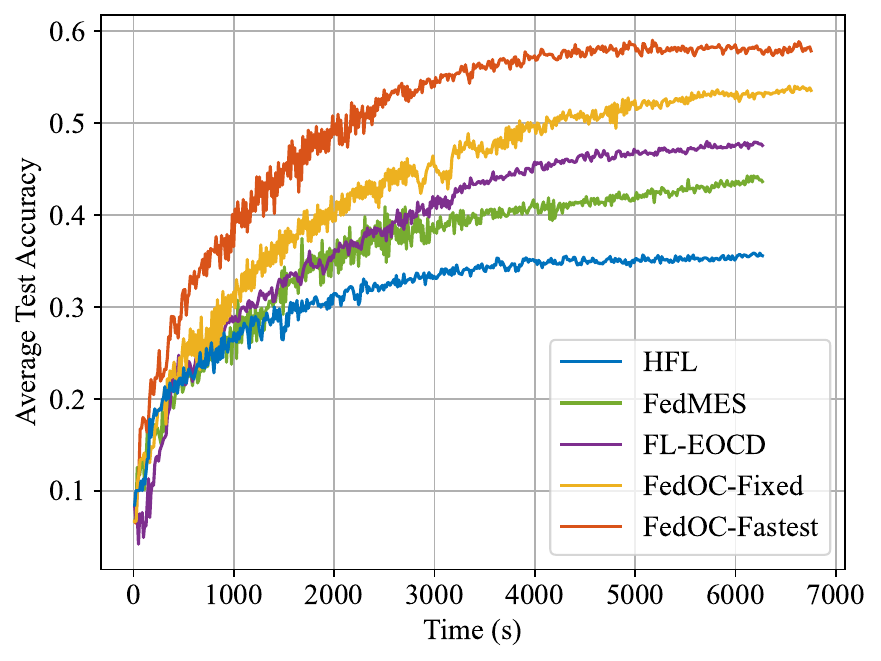}
    }
    \caption{Average test accuracy versus training time under a 4-ES deployment with \(|\mathcal{V}_{1,2}| = |\mathcal{V}_{2,3}| = 7\) setup.}
    \label{fig:NON-IID-7-overlap}
\end{figure*}

As shown in Fig.~\ref{fig:bars1} and Fig.~\ref{fig:bars2}, for HFL, the time to reach the target accuracy first decreases and then increases as $\kappa$ grows. This is because, within a moderate range, reducing the cloud aggregation frequency helps cut down communication overhead and thus accelerates convergence. However, as $\kappa$ continues to increase, the performance of HFL begins to deteriorate due to the lack of timely global synchronization, often resulting in slower convergence or even failure to converge when $\kappa$ becomes large. In contrast, our method significantly delays this trend. For MNIST, the inter-ES model transmissions in FedOC enable convergence before 250 rounds even without any cloud aggregation. As a result, increasing the cloud aggregation interval continuously reduces the convergence time in the range of \(\kappa \leq 250\), and the time to reach the target accuracy does not increase further even when $\kappa$ exceeds 250. For CIFAR, enlarging the cloud aggregation interval also reduces communication overhead and facilitates convergence. Moreover, due to inter-ES model transmissions, the model can still converge quickly even when $\kappa$ is large. Only when $\kappa = 1000$, the time to reach the target accuracy shows a slight increase. Although model transmission occurs only among neighboring edges each round, information from distant edges propagates over rounds. For example, in round~1, ES~1 receives the locally aggregated model from ES~2, which will aggregate models from ESs~1--3. In round~2,  ES~1 receives the updated model from ES 2, which already incorporates information from ES 3 in the previous round. By extension, after \(L-1\) rounds, the model at ES~1 contains information from all ESs. This neighbor ES exchange enables effective model propagation even under cloud-free aggregation. Based on these observations, we further investigate the cloud-free setting, comparing our approach with decentralized baselines and a cloud-free variant of HFL under different data distributions and overlapping configurations.

\subsection{Algorithm Evaluation}

Fig.~\ref{fig:NON-IID-1-overlap-new} presents the convergence results for the minimal overlap setting where only one ROC exists between each pair of overlapped cells. Here, the benefits of inter-ES models exchange by ROC relaying are pronounced. HFL and FedMES only perform intra-cell aggregation in each round, and as a result, their models plateau at a lower accuracies. FL-EOCD fares better: by integrating other cells’ models, albeit with a delay, it achieves higher accuracy than HFL/FedMES in this scenario. However, due to the staleness of the aggregated models, FL-EOCD’s advantage over FedMES diminishes on the harder CIFAR-10 task. In contrast, FedOC’s two mechanisms achieve the similar highest accuracy on both datasets, with particularly striking gains for CIFAR-10. By the end of 500 rounds, FedOC-Fastest attains a final CIFAR-10 accuracy approximately 9–12 percentage points higher than all baselines. Since there is only one OC, the two OC selection mechanisms exhibit negligible differences in model performance. The superiority of FedOC confirms that real-time inter-cell update relays are crucial especially when each cell’s data is biased.  

As shown in Fig.~\ref{fig:NON-IID-10-overlap}, when the OCs increase to 10 clients per region, the accuracy of FedMES increases significantly, benefiting from its reliance on multiple OCs to aggregate ES models, which indirectly enhances the generalization of each ES model. However, relying solely on the limited aggregation within OCs provides only constrained gains in model generalization, while also incurring delays as each OC must wait for all associated ESs to finish transmission. FedOC-Fixed alleviates this issue by introducing inter-ES model exchange via the ROC, enabling real-time transmission of model information among neighboring ESs and thereby indirectly obtaining model knowledge from other cells. On top of this, FedOC-Fastest allows OCs to begin training as soon as a single ES model is received, effectively removing waiting delays and accelerating round completion without degrading model quality. Consequently, it achieves the best performance on both datasets.

\subsubsection{Experiments with Multiple Overlapping ESs}
To further evaluate the scalability and robustness of FedOC, we extend our experiments to a four-ES topology. The corresponding results are presented in Fig.~\ref{fig:NON-IID-1-overlap-4cells} and Fig.~\ref{fig:NON-IID-7-overlap}. As shown in these figures, FedOC continues to deliver strong performance as the number of edge servers increases, even under highly skewed data distributions. The trends are consistent with those observed in the 3-ES setting: FedOC-Fastest achieves the fastest convergence and the highest accuracy, followed by FedOC-Fixed. In particular, as shown in Fig.~\ref{fig:NON-IID-7-cifar10-4cells}, FedOC-Fastest attains more than a 10\%~25\% accuracy improvement on CIFAR-10 compared with other three benchmarks, while the Fixed variant achieves a gain of over \(5\%\). These results highlight the robustness of FedOC and demonstrate the scalability of FedOC in larger networks.

\section{Conclusion}
In this paper, we presented FedOC, a novel federated learning architecture that leverages OCs as both communication relays and adaptive training participants to improve learning efficiency in wireless edge networks. FedOC supports round-wise model sharing across neighboring edge servers without reliance on a central server, thereby speeding up the global learning process. Besides that, we introduced a dynamic participating strategy for OCs that adaptively selects the first-received ES model as the initial training model of each round to minimize latency. We derived the theoretical convergence bound of FedMes and provided insights on the convergence behavior. Extensive experiments confirm that FedOC consistently outperforms classic FL algorithms and recent overlapping-area based methods on both convergence speed and accuracy. Our solution accelerates FL to support latency-sensitive applications in wireless networks especially characterized by high client–server delays and non-IID data distributions across edge servers.

\appendix
\subsection{Proof of Lemma 1}
Using the expression in \eqref{equ7}, we can rewrite the ES model in our FedOC framework as
\begin{align} 
\label{equ13}
\bm {w}_{r+1}^{(f_l)}
= \sum_{i=l-1}^{l+1}\frac{\hat{N}_{r}^{(f_i)} \bm{\hat{w}}_{r,E}^{(f_i)}}{\sum_{i=l-1}^{l+1} \hat{N}_{r}^{(f_i)}},
\end{align}
where \(\bm{\hat{w}}_{r,E}^{(f_i)}\) is denoted as:

\begin{align}
\label{equ14}
\bm{\hat{w}}_{r,E}^{(f_i)} = \sum_{k\in \mathcal{\hat{K}}_{r}^{(f_i)}}\frac{n^{(k)}}{\hat{N}_{r}^{(f_i)}} \bm{w}_{r,E}^{(k)},
\end{align}
and \(\hat{N}_{r}^{(f_i)} = \sum_{k\in \mathcal{\hat{K}}_{r}^{(f_i)}} n^{(k)}\), and based on the \eqref{equ5} and \eqref{equ6}, \(\hat{\mathcal{K}}_{r}^{(f_i)}\) satisfies:

\begin{align}
\label{equ15}
\hat{\mathcal{K}}_{r}^{(f_i)} = 
\begin{cases} 
\mathcal{S}_{l} \cup b_{1,2}, & l = 1\\
\mathcal{S}_{l}, & l = 2 \\
\mathcal{S}_{l} \cup b_{2,3}. & l=3. \\ 
\end{cases}
\end{align}
Under these definitions, each ROC is mathematically associated with a unique cell. Since the participating client sets and their data samples in each cell remain fixed across rounds, we drop the subscript \(r\) in 
\(\hat{\mathcal{K}}_{r}^{(f_l)}\) and \(\hat{N}_{r}^{(f_l)}\).

For FedOC, we assume the cloud aggregation is performed every \(\kappa\) rounds. let \(\bm{w}_{r}^{(f)}\) denote the cloud model, we have

\begin{align}
\label{equ16}
\bm{w}_{r}^{(f)} = \sum_{i=1}^{L}
\frac{\hat{N}^{(f_i)} \hat{\bm{w}}_{r-1,E}^{(f_i)}}{\sum_{i=1}^{L} \hat{N}^{(f_i)}}
\quad \text{for } \kappa \mid r.
\end{align}

For cell-centralized SGD, let \(\bm{w}_{r}^{(c)}\) denote the global model in round \(r\), and \(\bm{w}_{r,e}^{(c_l)}\) represent the model of cell \(l\) after \(e\) steps of local SGD under this centralized setting. The update rule of cell-centralized SGD is given by:
\begin{align} 
\label{equ17}
\bm{w}_{r,e+1}^{(c_l)} &= \bm{w}_{r,e}^{(c_l)} - \eta_{r,e}  \sum_{i=1}^C P_{y=i}^{(c_l)} \nabla_{\bm{w}} \mathbb{E}_{\boldsymbol{x} \mid y=i}\left[\log f_i(\boldsymbol{x}, \bm{w}_{r,e}^{(c_l)})\right],
\end{align}
where \(e \in \{0, 1, \ldots, E-1\}\), and \(P_{y=i}^{(c_l)}\) represents the population of data samples with label \(i\) across all clients in cell \(l\).  In the centralized-SGD, cloud aggregation occurs in each round, so the cloud model satisfies
\begin{align} 
\label{equ18}
\bm{w}_{r+1}^{(c)} =
\frac{\sum_{i=1}^{L} \hat{N}^{(f_i)}\, \bm{w}_{r,E}^{(c_i)}}
     {\sum_{i=1}^{L} \hat{N}^{(f_i)}}, \quad r=0,1, \ldots,R-1 .
\end{align}

Based on~\eqref{equ10}, (14)-(18), we have

%明天继续修改这个公式
\begin{align}
\label{equ19}
&\left\|\bm{w}_{R}^{(f)} - \bm{w}_{R}^{(c)} \right\|\notag\\
=& 
\bigg\| 
\sum_{j=1}^{L} 
\frac{\hat{N}^{(f_j)}}{N} 
\left(\bm{\hat{w}}_{R-1,E}^{(f_j)} - \bm{w}_{R-1,E}^{(c_j)} \right)
\bigg\| \notag\\
% =& \bigg\| 
% \sum_{j=1}^{L} 
% \frac{\hat{N}^{(f_j)}}{N} 
% \bigg[ \sum_{k\in \hat{\mathcal{K}}^{(f_j)}}
% \frac{n^{(k)}}{\hat{N}^{(f_j)}} 
% \bigg(\bm{w}_{R-1,E-1}^{(k)} - \eta_{R-1,E-1} \notag\\
% &\cdot \sum_{i=1}^C P_{y=i}^{(k)} \nabla_{\bm{w}} \mathbb{E}_{\boldsymbol{x} \mid y=i}\left[\log f_i(\boldsymbol{x}, \bm{w}_{R-1,E-1}^{(k)})\right] \bigg) - \bm{w}_{R-1,E-1}^{(c_j)} \notag\\
% &+ \eta_{R-1,E-1}  \sum_{i=1}^C P_{y=i}^{(c_j)} \nabla_{\bm{w}} \mathbb{E}_{\boldsymbol{x} \mid y=i}\left[\log f_i(\boldsymbol{x}, \bm{w}_{R-1,E-1}^{(c_j)})\right]
% \bigg] \bigg\| \notag\\
 \stackrel{1}=& \bigg\| 
\sum_{j=1}^{L} 
\frac{\hat{N}^{(f_j)}\sum_{k=1}^{\hat{K}_{r}^{(f_j)}} 
\frac{n^{(k)}}{\hat{N}^{(f_j)}} }{N} 
\left[ \left( \bm{w}_{R-1,E-1}^{(k)} - \bm{w}_{R-1,E-1}^{(c_j)} \right) \right. \notag\\
& + \left. \eta_{R-1,E-1} \sum_{i=1}^C P_{y=i}^{(k)} 
\left( \nabla_{\bm{w}} \mathbb{E}_{\boldsymbol{x} \mid y=i} \left[\log f_i(\boldsymbol{x}, \bm{w}_{R-1,E-1}^{(k)}) \right] \right. \right. \notag\\
& \left. \left. - \nabla_{\bm{w}} \mathbb{E}_{\boldsymbol{x} \mid y=i} \left[\log f_i(\boldsymbol{x}, \bm{w}_{R-1,E-1}^{(c_j)}) \right] \right)
\right] 
\bigg\| \notag\\
 \stackrel{2}\leq &
\sum_{j=1}^{L} 
\frac{\hat{N}^{(f_j)}}{\sum_{j=1}^{L} \hat{N}^{(f_j)}} 
\sum_{k=1}^{\hat{K}_{r}^{(f_j)}} 
\frac{n^{(k)}}{\hat{N}^{(f_j)}} 
\left(
1 + a_{R-1,E-1}^{(k)} \right) \notag \\
& \cdot \left\|
\bm{w}_{R-1,E-1}^{(k)} - \bm{w}_{R-1,E-1}^{(c_j)} \right\|,
\end{align}
where equality 1 holds because for each class \(i \in \{1, \ldots, C\}\), \( P_{y=i}^{(c_j)} = \sum_{k=1}^{\hat{K}_{r}^{(f_j)}} 
\frac{n^{(k)}}{\hat{N}^{(f_j)}} P_{y=i}^{(k)}\). Inequality 2 holds because of \textbf{Assumption~ \ref{assumption1}} and we denote \(a_{r,q}^{(k)}=1+\eta_{r,q} \sum_{i=1}^C P_{y=i}^{(k)} \lambda_{\boldsymbol{x} \mid y=i}\). 

Inspired by \cite{a34}, we bound the discrepancy between the local model of client \(k \in \hat{\mathcal{K}}^{(c_j)}\) and the corresponding ES model as follows: 

\begin{align} 
\label{equ20}
&\bigg\| \boldsymbol{w}_{R-1,E-1}^{(k)} - \boldsymbol{w}_{R-1,E-1}^{(c_j)} \bigg\| \notag\\
% \stackrel{3}{=}& \bigg\| \left(\boldsymbol{w}_{R-1,E-2}^{(k)} -  \boldsymbol{w}_{R-1,E-2}^{(c_j)} \right) \notag\\
% & + \bigg( \eta_{R-1,E-2} \sum_{i=1}^C P_{y=i}^{(k)} \nabla_{\boldsymbol{w}} \mathbb{E}_{\boldsymbol{x} \mid y=i} \left[ \log f_i(\boldsymbol{x}, \boldsymbol{w}_{R-1,E-2}^{(c_j)}) \right] \notag\\
% & - \eta_{R-1,E-2} \sum_{i=1}^C P_{y=i}^{(k)} \nabla_{\boldsymbol{w}} \mathbb{E}_{\boldsymbol{x} \mid y=i} \left[ \log f_i(\boldsymbol{x}, \boldsymbol{w}_{R-1,E-2}^{(k)}) \right] \bigg) \bigg\| \notag\\
\stackrel{3}=& \bigg\| \boldsymbol{w}_{R-1,E-2}^{(k)} - \boldsymbol{w}_{R-1,E-2}^{(c_j)} \bigg\| \notag\\
& + \eta_{R-1,E-2} \bigg\| \sum_{i=1}^C P_{y=i}^{(k)} \nabla_{\boldsymbol{w}} \mathbb{E}_{\boldsymbol{x} \mid y=i} \left[ \log f_i(\boldsymbol{x}, \boldsymbol{w}_{R-1,E-2}^{(k)}) \right] \notag\\
& - \sum_{i=1}^C P_{y=i}^{(k)} \nabla_{\boldsymbol{w}} \mathbb{E}_{\boldsymbol{x} \mid y=i} \left[ \log f_i(\boldsymbol{x}, \boldsymbol{w}_{R-1,E-2}^{(c_j)}) \right] \notag\\
& + \sum_{i=1}^C P_{y=i}^{(k)} \nabla_{\boldsymbol{w}} \mathbb{E}_{\boldsymbol{x} \mid y=i} \left[ \log f_i(\boldsymbol{x}, \boldsymbol{w}_{R-1,E-2}^{(c_j)}) \right] \notag\\
& - \sum_{i=1}^C P_{y=i}^{(c_j)} \nabla_{\boldsymbol{w}} \mathbb{E}_{\boldsymbol{x} \mid y=i} \left[ \log f_i(\boldsymbol{x}, \boldsymbol{w}_{R-1,E-2}^{(c_j)}) \right] \bigg\| \notag\\
\leq& \bigg\| \boldsymbol{w}_{R-1,E-2}^{(k)} - \boldsymbol{w}_{R-1,E-2}^{(c_j)} \bigg\| \notag\\
& + \eta_{R-1,E-2} \bigg[ \bigg\| \sum_{i=1}^C P_{y=i}^{(k)} \bigg( \nabla_{\boldsymbol{w}} \mathbb{E}_{\boldsymbol{x} \mid y=i} \left[ \log f_i(\boldsymbol{x}, \boldsymbol{w}_{R-1,E-2}^{(k)}) \right] \notag\\
& - \nabla_{\boldsymbol{w}} \mathbb{E}_{\boldsymbol{x} \mid y=i} \left[ \log f_i(\boldsymbol{x}, \boldsymbol{w}_{R-1,E-2}^{(c_j)}) \right] \bigg) \bigg\| \notag\\
& + \bigg\| \sum_{i=1}^C \left( P_{y=i}^{(k)} - P_{y=i}^{(c_j)} \right)  \nabla_{\boldsymbol{w}} \mathbb{E}_{\boldsymbol{x} \mid y=i} \left[ \log f_i(\boldsymbol{x}, \boldsymbol{w}_{R-1,E-2}^{(c_j)}) \right] \bigg\| \bigg] \notag\\
\stackrel{4}{\leq}& \bigg(1 + \eta_{R-1,E-2} \sum_{i=1}^C P_{y=i}^{(k)} \lambda_{x\mid y=i}\bigg) \bigg\| \boldsymbol{w}_{R-1,E-2}^{(k)} - \boldsymbol{w}_{R-1,E-2}^{(c_j)} \bigg\| \notag\\
& + \eta_{R-1,E-2} \sum_{i=1}^C \bigg| P_{y=i}^{(k)} - P_{y=i}^{(c_j)} \bigg|  g_{\max} (\boldsymbol{w}_{R-1,E-2}^{(c_j)}) \notag\\
=& a_{R-1,E-2}^{(k)} \bigg\| \boldsymbol{w}_{R-1,E-2}^{(k)} - \boldsymbol{w}_{R-1,E-2}^{(c_j)} \bigg\| \notag\\
& +  \eta_{R-1,E-2} \sum_{i=1}^C \bigg| P_{y=i}^{(k)} - P_{y=i}^{(c_j)} \bigg| g_{\max} (\boldsymbol{w}_{R-1,E-2}^{(c_j)}),
\end{align}
where equality 3 holds because the SGD update rules in~\eqref{equ10} and~\eqref{equ17}. Inequality 4 holds because the \textbf{Assumption 1} and we denote \(g_{\max }\left(\boldsymbol{w}_{r, q}^{(c_j)}\right)=\max _{i=1}^C\left\|\nabla_{\boldsymbol{w}} \mathbb{E}_{\boldsymbol{x} \mid y=i} \log f_i\left(\boldsymbol{x}, \boldsymbol{w}_{r, q}^{(c_j)}\right)\right\|\) for \(r \in \{1,2,\ldots,R\}\) and \(q \in \{0,1,\ldots, E-1\}\). From~\eqref{equ20}, it follows that

\begin{align}
\label{equ21}
&\bigg\| \boldsymbol{w}_{R-1,E-1}^{(k)} - \boldsymbol{w}_{R-1,E-1}^{(c_j)} \bigg\| \notag \\
\leq& a_{R-1,E-2}^{(k)} \bigg\| \boldsymbol{w}_{R-1,E-2}^{(k)} - \boldsymbol{w}_{R-1,E-2}^{(c_j)} \bigg\| \notag \\
& + \sum_{i=1}^C \bigg| P_{y=i}^{(k)} - P_{y=i}^{(c_j)} \bigg| \eta_{R-1,E-2} g_{\max} (\boldsymbol{w}_{R-1,E-2}^{(c_j)}) \notag \\
% \leq & a_{R-1,E-2}^{(k)} a_{R-1,E-3}^{(k)} \bigg\| \boldsymbol{w}_{R-1,E-3}^{(k)} - \boldsymbol{w}_{R-1,E-3}^{(c_j)} \bigg\| \notag\\
% & + \eta_{R-1,E-2} \sum_{i=1}^C \bigg| P_{y=i}^{(k)} - P_{y=i}^{(c_j)} \bigg|  g_{\max} (\boldsymbol{w}_{R-1,E-2}^{(c_j)}) \notag\\
% & + \eta_{R-1,E-3} \sum_{i=1}^C \bigg| P_{y=i}^{(k)} - P_{y=i}^{(c_j)} \bigg|  a_{R-1,E-2}^{(k)} g_{\max} (\boldsymbol{w}_{R-1,E-3}^{(c_j)}) \notag\\
\leq& \left( \prod_{e=0}^{E-2} a_{R-1,e}^{(k)}\right)
\left\|\boldsymbol{w}_{R-1, 0}^{(k)}-\boldsymbol{w}_{R-1, 0}^{(c_j)}\right\| + \sum_{i=1}^C\left|P_{y=i}^{(k)}-P_{y=i}^{(c_j)}\right|  \notag\\
&\cdot \sum_{e=0}^{E-2}  \eta_{R-1,e}  \left(\prod_{d=e+1}^{E-2}a_{R-1,d}^{(k)}\right) g_{\max }\left(\boldsymbol{w}_{R-1,e}^{(c_j)}\right),
\end{align}
where we define \(\prod_{i = a}^{b} D_i = 1\) when \(a > b\), following standard convention. Substituting \eqref{equ21} into \eqref{equ19}, we have

\begin{align} 
\label{equ22}
&\left\|\bm{w}_{R}^{(f)} - \bm{w}_{R}^{(c)} \right\| \notag\\
\leq& \sum_{j=1}^{L} 
\frac{\hat{N}^{(f_j)}}{\sum_{j=1}^{L} \hat{N}^{(f_j)}} 
\sum_{k \in \hat{\mathcal{K}}^{(f_j)}} 
\frac{n^{(k)}}{\hat{N}^{(f_j)}} 
\Bigg[ \left( \prod_{e=0}^{E-1} a_{R-1,e}^{(k)}\right) \notag\\ 
\quad& \cdot \left\|\boldsymbol{w}_{R-1, 0}^{(f_j)}-\boldsymbol{w}_{R-1, 0}^{(c)}\right\| + \sum_{i=1}^C\left|P_{y=i}^{(k)}-P_{y=i}^{(c_j)}\right| \notag\\
\quad& \cdot \sum_{e=0}^{E-2}  \eta_{R-1,e}  \left(\prod_{d=e+1}^{E-1} a_{R-1,d}^{(k)}\right) g_{\max }\left(\boldsymbol{w}_{R-1,e}^{(c_j)}\right) \Bigg]  \notag\\
= &\sum_{j=1}^{L} 
\bigg( D_{R-1}^{(j)}  \left\|\boldsymbol{w}_{R-1,0}^{(f_j)}-\boldsymbol{w}_{R-1,0}^{(c)}\right\| + G_{R-1}^{(j)} \bigg),\end{align}
where we denote

\begin{align}
\label{equ23}
D_{R-1}^{(j)}
&= \frac{\hat{N}^{(f_j)}}{\sum_{u=1}^{L}\hat{N}^{(f_u)}}
   \sum_{k \in \hat{\mathcal{K}}^{(f_j)}} 
   \frac{n^{(k)}}{\hat{N}^{(f_j)}} 
   \left(\prod_{e=0}^{E-1} a_{R-1,e}^{(k)}\right),
\end{align}

\begin{align}
\label{equ24}
G_{R-1}^{(j)}
&= \sum_{e=0}^{E-2} \eta_{R-1,e}\,\beta_{R-1,e}^{(j)},
\end{align}
and
\begin{align}
\label{equ25}
\beta_{R-1,e}^{(j)}=& \frac{\hat{N}^{(f_j)}}{\sum_{u=1}^{L}\hat{N}^{(f_u)}}
   \sum_{k \in \hat{\mathcal{K}}^{(f_j)}} 
   \frac{n^{(k)}}{\hat{N}^{(f_j)}}
   \left(\sum_{i=1}^C \bigl|P_{y=i}^{(k)} - P_{y=i}^{(c_j)}\bigr|\right) \notag\\
   &\cdot \left(\prod_{d=e+1}^{E-1} a_{R-1,d}^{(k)}\right)
   g_{\max}\!\bigl(\boldsymbol{w}_{R-1,e}^{(c_j)}\bigr).
\end{align}

To bound \(\left\|\bm{w}_{R}^{(f)} - \bm{w}_{R}^{(c)} \right\|\), we need to give the bound of \( \left\|\boldsymbol{w}_{R-1, 0}^{(f_j)}-\boldsymbol{w}_{R-1, 0}^{(c)}\right\| \) for \(j = 1,\ldots,L\). Based on (13)-(15) and (18), we have:

\begin{align} 
\label{equ26}
&\left\| \boldsymbol{w}_{R-1, 0}^{(f_1)} - \boldsymbol{w}_{R-1, 0}^{(c)} \right\| \notag \\
= &\bigg\| \frac{\hat{N}^{(f_1)} \boldsymbol{\hat{w}}_{R-2, E}^{(f_1)} + \hat{N}^{(f_2)} \boldsymbol{\hat{w}}_{R-2, E}^{(f_2)}}{\hat{N}^{(f_1)} + \hat{N}^{(f_2)}} - \frac{\sum_{j=1}^{3} \hat{N}^{(f_j)} \boldsymbol{w}_{R-2, E}^{(c_j)}}{\sum_{j=1}^{3} \hat{N}^{(f_j)}}\bigg\| \notag \\
% = &\bigg\| \sum_{j=1}^{3} \frac{\hat{N}^{(f_j)} \left( \boldsymbol{\hat{w}}_{R-2, E}^{(f_j)} - \boldsymbol{w}_{R-2, E}^{(c_j)} \right) }{\sum_{u=1}^{3} \hat{N}^{(f_u)}} \notag \\
% &+ \left[\frac{\hat{N}^{(f_1)}}{\sum_{j=1}^{2} \hat{N}^{(f_j)}} - \frac{\hat{N}^{(f_1)}}{\sum_{j=1}^{3} \hat{N}^{(f_j)}}\right] \boldsymbol{\hat{w}}_{R-2, E}^{(f_1)} \notag \\ 
% &+ \left[\frac{\hat{N}^{(f_2)}}{\sum_{j=1}^{2} \hat{N}^{(f_j)}} - \frac{\hat{N}^{(f_2)}}{\sum_{j=1}^{3} \hat{N}^{(f_j)}}\right]  \boldsymbol{\hat{w}}_{R-2, E}^{(f_2)} \notag \\
% &- \frac{\hat{N}^{(f_3)} \boldsymbol{\hat{w}}_{R-2, E}^{(f_3)}}{\sum_{j=1}^{3} \hat{N}^{(f_j)}} \bigg\| \notag \\
=& \bigg\| 
\sum_{j=1}^{3} \frac{\hat N^{(f_j)}\!\left(\hat{\boldsymbol w}_{R-2,E}^{(f_j)}-\boldsymbol w_{R-2,E}^{(c_j)}\right)}{\sum_{u=1}^{3}\hat N^{(f_u)}}
\;+\; \sum_{j=1}^{3} \rho_{R-2}^{(j)}\,\hat{\boldsymbol w}_{R-2,E}^{(f_j)}
\bigg\|,
\end{align}
where 
\begin{align}
\label{equ27}
\rho_{R-2}^{(j)} =
\begin{cases}
\dfrac{\hat N^{(f_j)}}{\hat N^{(f_1)}+\hat N^{(f_2)}}-
\dfrac{\hat N^{(f_j)}}{\sum_{u=1}^{3}\hat N^{(f_u)}}, & j=1,2,\\[6pt]
-\dfrac{\hat N^{(f_3)}}{\sum_{u=1}^{3}\hat N^{(f_u)}}, & j=3,
\end{cases}
\end{align}
and we have \(\sum_{j=1}^{3}\rho_{r}^{(j)} = 0\), \(\forall r\). By using the same way, we can obtain the bound of other edge models as follows:
\begin{align} 
\label{equ28}
&\bigg\| \boldsymbol{w}_{R-1, 0}^{(f_2)} - \boldsymbol{w}_{R-1, 0}^{(c)} \bigg\| \notag \\
\leq & \bigg\| \sum_{j=1}^{3} \frac{\hat{N}^{(f_j)} \left( \boldsymbol{\hat{w}}_{R-2, E}^{(f_j)} - \boldsymbol{w}_{R-2, E}^{(c_j)} \right) }{\sum_{u=1}^{3} \hat{N}^{(f_u)}}\bigg\|,
\end{align}

\begin{align} 
\label{equ29}
&\bigg\| \boldsymbol{w}_{R-1, 0}^{(f_3)} - \boldsymbol{w}_{R-1, 0}^{(c)} \bigg\| \notag\\
\leq & \bigg\| \sum_{j=1}^{3} \frac{\hat{N}^{(f_j)} \left( \boldsymbol{\hat{w}}_{R-2, E}^{(f_j)} - \boldsymbol{w}_{R-2, E}^{(c_j)} \right) }{\sum_{u=1}^{3} \hat{N}^{(f_u)}} + \sum_{j=1}^{3}\mu_{R-2}^{(j)}\boldsymbol{\hat{w}}_{R-2, E}^{(f_j)}\bigg\|,
\end{align}
where 
\begin{align}
\label{equ30}
\mu_{R-2}^{(j)} =
\begin{cases}
-\dfrac{\hat N^{(f_1)}}{\sum_{u=1}^{3}\hat N^{(f_u)}}, & j=1,\\[6pt]
\dfrac{\hat N^{(f_j)}}{\hat N^{(f_2)}+\hat N^{(f_3)}}-
\dfrac{\hat N^{(f_j)}}{\sum_{u=1}^{3}\hat N^{(f_u)}}, & j=2,3,
\end{cases}
\end{align}
and we have \(\sum_{j=1}^{3}\mu_{r}^{(j)}=0\), \(\forall r\). By substituting~\eqref{equ26}–\eqref{equ29} into~\eqref{equ22} we have:

\begin{align} 
\label{equ31}
&\left\|\bm{w}_{R}^{(f)} - \bm{w}_{R}^{(c)} \right\| \notag\\
% =& \bigg\| \sum_{j=1}^{3} \frac{\hat{N}^{(f_j)} \left( \boldsymbol{\hat{w}}_{R-1, E}^{(f_j)} - \boldsymbol{w}_{R-1, E}^{(c_j)} \right) }{\sum_{u=1}^{3} \hat{N}^{(f_u)}} \bigg\| \notag\\
% \leq& \sum_{j=1}^{3} 
% \bigg( D_{R-1}^{(j)}  \left\|\boldsymbol{w}_{R-1, 0}^{(f_j)}-\boldsymbol{w}_{R-1, 0}^{(c)}\right\| + G_{R-1}^{(j)} \bigg) \notag \\
\leq& D_{R-1} \bigg\| \sum_{j=1}^{3} \frac{\hat{N}^{(f_j)} \left( \boldsymbol{\hat{w}}_{R-2, E}^{(f_j)} - \boldsymbol{w}_{R-2, E}^{(c_j)} \right) }{\sum_{u=1}^{3} \hat{N}^{(f_u)}} \bigg\| + \epsilon_{R-1}^{\text{intra}} + \epsilon_{R-1}^{\text{inter}} \notag\\
% \leq& \prod_{p=R-\kappa +1}^{R-1} D_{p} \bigg\| \sum_{j=1}^{3} \frac{\hat{N}^{(f_j)} \left( \boldsymbol{\hat{w}}_{R-\kappa, E}^{(f_j)} - \boldsymbol{w}_{R-\kappa, E}^{(c_j)} \right) }{\sum_{u=1}^{3} \hat{N}^{(f_u)}} \bigg\| + \notag \\
% &\sum_{p=R-\kappa+1}^{R-1}\left(\prod_{q=p+1}^{R-1} D_q \right)\epsilon_{p}^{\text{intra}} + \sum_{p=R-\kappa+1}^{R-1}\left(\prod_{q=p+1}^{R-1} D_q\right)\epsilon_{p}^{\text{inter}} \notag\\
\leq& \prod_{p=R-\kappa +1}^{R-1} D_{p} \left( \sum_{j=1}^{3} D_{R-\kappa}^{(j)}  \left\|\boldsymbol{w}_{R-\kappa, 0}^{(f)}-\boldsymbol{w}_{R-\kappa, 0}^{(c)}\right\| + \epsilon_{R-\kappa}^{\text{intra}}\right)+\notag \\
&\sum_{p=R-\kappa+1}^{R-1}\left(\prod_{q=p+1}^{R-1} D_q \right)\epsilon_{p}^{\text{intra}} + \sum_{p=R-\kappa+1}^{R-1}\left(\prod_{q=p+1}^{R-1} D_q\right)\epsilon_{p}^{\text{inter}},
\end{align}
where the first inequality holds since we denote \(D_{R-1} = \sum_{j=1}^{3} D_{R-1}^{(j)}\), 
\begin{align}
\label{equ32}
\epsilon_{R-1}^{\text{intra}} = \sum_{j=1}^{3} G_{R-1}^{(j)},
\end{align}
and
\begin{align}
\label{equ33}
\epsilon_{R-1}^{\text{inter}} =& D_{R-1}^{(1)} \bigg\| \sum_{j=1}^{3}\rho_{R-2}^{(j)} \boldsymbol{\hat{w}}_{R-2, E}^{(f_j)}\bigg\| \notag\\
&+D_{R-1}^{(3)} \bigg\| \sum_{j=1}^{3} \mu_{R-2}^{(j)}\boldsymbol{\hat{w}}_{R-2, E}^{(f_j)}\bigg\|.
\end{align}

We assume the initial model of FedOC and the centralized-SGD algorithm after cloud aggregation is the same, i.e., \(\left\|\boldsymbol{w}_{R-\kappa, 0}^{(f)}-\boldsymbol{w}_{R-\kappa, 0}^{(c)}\right\| = 0\) then we have

\begin{align}
\label{equ34}
\left\|\bm{w}_{R}^{(f)} - \bm{w}_{R}^{(c)} \right\|
\leq&\underbrace{\sum_{p=R-\kappa}^{R-1}\left(\prod_{q=p+1}^{R-1} D_q \right)\epsilon_{p}^{\text{intra}}}_{\epsilon^{\text{intra}}} \notag\\
&+ \underbrace{\sum_{p=R-\kappa+1}^{R-1}\left(\prod_{q=p+1}^{R-1} D_q\right)\epsilon_{p}^{\text{inter}}}_{\epsilon^{\text{inter}}},
\end{align}
which completes the proof.

\subsection{Proof of Theorem 1}

Based on the \textbf{Assumption 1}, we have
\begin{align} 
\label{equ35}
\ell(\bm{w}_{R}^{(f)}) - \ell(\bm{w}^{*}) &\leq \frac{\lambda}{2} \left\| \bm{w}_{R}^{(f)} - \bm{w}^{*} \right\| \notag\\
% & = \frac{\lambda}{2} \left\| \bm{w}_{R}^{(f)} - \bm{w}_{R}^{(c)} + \bm{w}_{R}^{(c)} - \bm{w}^{*} \right\| \notag\\
&\leq \frac{\lambda}{2} \underbrace{\left\| \bm{w}_{R}^{(f)} - \bm{w}_{R}^{(c)} \right\|}_{A_1} + \frac{\lambda}{2} \underbrace{\left\| \bm{w}_{R}^{(c)} - \bm{w}^{*} \right\|}_{A_2} 
\end{align}
where \(A_1\) has been bounded in \textbf{Lemma 1}, we need to further specify the bounds of \(\epsilon^{\text{intra}}\) and \(\epsilon^{\text{inter}}\) by combining the assumption of \(\eta_{r,e} = \frac{1}{r(E -1)}\).

In the following, we first derive a theoretical upper bound for \(\epsilon^{\text{intra}}\). We assume that there exist constants \(g_{\max}\) and \(\beta_{r}^{(j)}\) such that \(g_{\max}\big(\bm{w}_{r,e}^{(c_j)}\big) \le g_{\max}\) and \(\beta_{r,e}^{(j)} \le \beta_{r}^{(j)}\), \(\forall r,~e,~j\). Moreover, we denote \(\sum_{j=1}^{L} \beta_{r}^{(j)} = \beta_{r}\). Since \(a_{r,q}^{(k)}=1+\eta_{r,q} \sum_{i=1}^C P_{y=i}^{(k)} \lambda_{\boldsymbol{x} \mid y=i} >1,~\forall r,~q\), based on \eqref{equ25}, we further assume \(\beta_{\text{min}} \leq \beta_{r} \leq \beta_{\text{max}}\). By substituting \(\eta_{r,e} = \frac{1}{r(E-1)}\), \eqref{equ24} and \eqref{equ25} into \eqref{equ32}, we have

\begin{align} 
\label{equ36}
\epsilon_{r}^{\text{intra}} = \sum_{e=0}^{E2}\eta_{r,e}\sum_{j=1}^{L}\beta_{r,e}^{(j)}  
\leq\sum_{e=0}^{E-2} \eta_{r,e} \beta_{r} \leq   \frac{\beta_{r}}{r},
\end{align}

Substituting \eqref{equ36} into \(\epsilon^{\text{intra}}\) in \eqref{equ34}, it follows that
\begin{align} 
\label{equ37}
\epsilon^{\text{intra}} &= \sum_{p=R-\kappa}^{R-1} \left( \prod_{q=p+1}^{R-1} D_q \right) \epsilon_p^{\text{intra}} \notag\\
& \leq \frac{\beta_{R-1}}{R-1} + D_{R-1} \frac{\beta_{R-2}}{R-2} + \cdots +\prod_{p=R-\kappa+1}^{R-1} D_{p} \frac{\beta_{R-\kappa}}{R-\kappa} \notag \\
& \leq \frac{\kappa \beta_{\text{max}} \prod_{p =R-\kappa+1}^{R-1} D_p}{R-\kappa},
\end{align}
where the final inequality holds since \(D_r>1, \forall r\) with the definition of \eqref{equ21}.

Next, we derive a theoretical upper bound for \(\epsilon^{\text{inter}}\) in \eqref{equ34}. Based on \eqref{equ3} and \eqref{equ14}, we have

\begin{align}
\label{equ38}
\bigg\|\sum_{j=1}^{3} \rho_{R-2}^{(j)}\,\hat{\boldsymbol w}_{R-2,E}^{(f_j)} \bigg\|=&\bigg\| \sum_{j=1}^{3}\rho_{R-2}^{(j)} \frac{\sum_{k\in\hat{\mathcal{K}}^{(f_j)}}n^{(k)}}{\hat{N}^{(f_j)}}
\Big[\boldsymbol{w}_{R-2,0}^{(k)} \notag\\ &- \sum_{e=0}^{E-1} \eta_{R-2,e}\,\nabla \ell_k\!\big(\boldsymbol{w}^{(k)}_{R-2,e}\big)\Big] \bigg\|.
\end{align}

Since \(\boldsymbol{w}^{(k)}_{R-2,0} = \boldsymbol{w}^{(f_j)}_{R-2,0}, ~\text{for}~ k\in\hat{\mathcal{K}}^{(f_j)}\), substituting \(\eta_{r,e} = \frac{1}{r(E -1)}\) into \eqref{equ38}, we have 
\begin{align} 
\label{equ39}
&\bigg\| \sum_{j=1}^{3} \rho_{R-2}^{(j)}\,\hat{\boldsymbol w}_{R-2,E}^{(f_j)} \bigg\| \notag\\
=& \bigg\| \sum_{j=1}^{3}\rho_{R-2}^{(j)}\boldsymbol{w}_{R-2,0}^{(f_j)} \bigg\| + \frac{1}{(R-2)(E-1)} \notag\\ 
&\cdot  \bigg\|\sum_{j=1}^{3}\rho_{R-2}^{(j)}\sum_{k\in\hat{\mathcal{K}}^{(f_j)}}\frac{n^{(k)}}{\hat{N}^{(f_j)}} \sum_{e=0}^{E-1} \nabla \ell_k(\boldsymbol{w}^{(k)}_{R-2,e}) \bigg\|\notag\\
=& \bigg\| \rho_{R-2}^{(1)} \frac{\sum_{j=1}^{2} \hat{N}^{(f_j)} \boldsymbol{\hat{w}}_{R-3, E}^{(f_j)}}{\hat{N}^{(f_1)} + \hat{N}^{(f_2)}} + \rho_{R-2}^{(2)} \frac{\sum_{j=1}^{3}\hat{N}^{(f_j)} \boldsymbol{\hat{w}}_{R-3, E}^{(f_j)}  }{\sum_{u=1}^{3} \hat{N}^{(f_u)}} \notag\\
&+ \rho_{R-2}^{(3)} \frac{\sum_{j=2}^{3} \hat{N}^{(f_j)} \boldsymbol{\hat{w}}_{R-3, E}^{(f_j)}}{\hat{N}^{(f_2)} + \hat{N}^{(f_3)}}   \bigg\| + 
\frac{\left\|\sum_{j=1}^{3} \rho_{R-2}^{(j)}\bar{\Delta}_{R-2}^{(j)}\right\|}{(R-2)(E-1)}\notag\\
=& \bigg\| \sum_{j=1}^{3}\rho_{R-3}^{(j)}\hat{\boldsymbol{w}}_{R-3,E}^{(f_j)}\bigg\| + \frac{\left\|\sum_{j=1}^{3} \rho_{R-2}^{(j)}\bar{\Delta}_{R-2}^{(j)}\right\|}{(R-2)(E-1)} \notag\\
=&  \bigg\| \sum_{j=1}^{3}\rho_{R-\kappa}^{(j)}\hat{\boldsymbol{w}}_{R-\kappa,0}^{(f_j)}\bigg\| + \sum_{u = R-\kappa}^{R-2} \frac{\left\|\sum_{j=1}^{3} \rho_{u}^{(j)}\bar{\Delta}_{u}^{(j)}\right\|}{u(E-1)} \notag\\
= & \bigg\| (\sum_{j=1}^{3}\rho_{R-\kappa}^{(j)})\boldsymbol{w}_{R-\kappa}^{(f)}\bigg\| + \sum_{u = R-\kappa}^{R-2} \frac{ \left\| \sum_{j=1}^{3} \rho_{u}^{(j)}\bar{\Delta}_{u}^{(j)} \right\|}{u(E-1)} \notag\\
=& \sum_{u = R-\kappa}^{R-2} \frac{ \left\| \sum_{j=1}^{3} \rho_{u}^{(j)} \bar{\Delta}_{u}^{(j)} \right\|}{u(E-1)}
\end{align}
where \(\bar{\Delta}_{R-2}^{(j)}\) is the cumulative average local gradients of cell \(j\), which is denoted by \(\bar{\Delta}_{R-2}^{(j)}=\sum_{k\in\hat{\mathcal{K}}^{(f_j)}}\frac{n^{(k)}}{\hat{N}^{(f_j)}} \sum_{e=0}^{E-1} \nabla \ell_k(\boldsymbol{w}^{(k)}_{R-2,e})\), and for \(r\leq R-2\), we denote
\begin{align}
\label{equ40}
\begin{cases}
\rho_{r-1}^{(1)}
= \dfrac{\hat N^{(f_1)}}{\hat N^{(f_1)}+\hat N^{(f_2)}}\,\rho_{r}^{(1)}
+ \dfrac{\hat N^{(f_1)}}{\sum_{u=1}^{3}\hat N^{(f_u)}}\,\rho_{r}^{(2)},\\
\rho_{r-1}^{(2)}
= \dfrac{\hat N^{(f_2)}\rho_{r}^{(1)}}{\sum_{j=1}^{2}\hat N^{(f_j)}}
+ \dfrac{\hat N^{(f_2)}\rho_{r}^{(2)}}{\sum_{u=1}^{3}\hat N^{(f_u)}}
+ \dfrac{\hat N^{(f_2)}\rho_{r}^{(3)}}{\sum_{j=2}^{3}\hat N^{(f_j)}},\\
\rho_{r-1}^{(3)}
= \dfrac{\hat N^{(f_3)}}{\sum_{u=1}^{3}\hat N^{(f_u)}}\,\rho_{r}^{(2)}
+ \dfrac{\hat N^{(f_3)}}{\hat N^{(f_2)}+\hat N^{(f_3)}}\,\rho_{r}^{(3)}.\\
\end{cases}
\end{align}
It is easily proved that \(\sum_{j=1}^{3} \rho_{r}^{(j)} = 0\), \(\forall~r\). Similarly, we have 
\begin{align} 
\label{equ41}
\bigg\|\sum_{j=1}^{3} \mu_{R-2}^{(j)}\,\hat{\boldsymbol w}_{R-2,E}^{(f_j)} \bigg\|= \sum_{p = R-\kappa}^{R-2} \frac{\left\|\sum_{j=1}^{3} \mu_{p}^{(j)}\bar{\Delta}_{p}^{(j)}\right\|}{p(E-1)},
\end{align}
and 
\begin{align}
\label{equ42}
\begin{cases}
\mu_{r-1}^{(1)}
= \dfrac{\hat N^{(f_1)}}{\hat N^{(f_1)}+\hat N^{(f_2)}}\,\mu_{r}^{(1)}
+ \dfrac{\hat N^{(f_1)}}{\sum_{u=1}^{3}\hat N^{(f_u)}}\,\mu_{r}^{(2)},\\
\mu_{r-1}^{(2)}
= \dfrac{\hat N^{(f_2)}\mu_{r}^{(1)}}{\sum_{j=1}^{2}\hat N^{(f_j)}}
+ \dfrac{\hat N^{(f_2)}\mu_{r}^{(2)}}{\sum_{u=1}^{3}\hat N^{(f_u)}}
+ \dfrac{\hat N^{(f_2)}\mu_{r}^{(3)}}{\sum_{j=2}^{3}\hat N^{(f_j)}},\\
\mu_{r-1}^{(3)}
= \dfrac{\hat N^{(f_3)}}{\sum_{u=1}^{3}\hat N^{(f_u)}}\,\mu_{r}^{(2)}
+ \dfrac{\hat N^{(f_3)}}{\hat N^{(f_2)}+\hat N^{(f_3)}}\,\mu_{r}^{(3)},\\
\end{cases}
\end{align}
where \(\sum_{j=1}^{3} \mu_{r}^{(j)} = 0,~\forall~r\) . Substituting \eqref{equ39} and \eqref{equ41} into \eqref{equ33}, we have

\begin{align}
\label{equ43}
\epsilon_{R-1}^{\text{inter}} &= {D_{R-1}^{(1)}}\left(\sum_{p = R-\kappa}^{R-2} \frac{\left\|\sum_{j=1}^{3} \rho_{p}^{(j)}\bar{\Delta}_{p}^{(j)}\right\|}{p(E-1)}\right) \notag\\
&\quad+{D_{R-1}^{(3)}} \left( \sum_{p = R-\kappa}^{R-2} \frac{\left\|\sum_{j=1}^{3} \mu_{p}^{(j)}\bar{\Delta}_{p}^{(j)}\right\|}{p(E-1)}\right) \notag\\
&=\sum_{p= R-\kappa}^{R-2} \frac{\sum_{j=1}^{3}\left(D_{r-1}^{(1)} \left\| \rho_{p}^{(j)}\right\| + D_{r-1}^{(3)} \left\| \mu_{p}^{(j)}\right\|\right) \bar{\Delta}_{p}^{(j)}}{p(E-1)}.
\end{align}

Based on \textbf{Assumption 3}, we have \(\left\|\bar{\Delta}_{p}^{(j)}\right\| \leq {\Delta}_{\text{max}}\). Assume there exist constants \(D_{\max}\) such that \( \sum_{j=1}^{3}\left(D_{r-1}^{(1)} \left\| \rho_{p}^{(j)}\right\| + D_{r-1}^{(3)} \left\| \mu_{p}^{(j)}\right\|\right) \leq D_{\text{max}}\) for \( R-\kappa \leq p \leq R-2\). Then we have

\begin{align}
\label{equ44}
\epsilon_{R-1}^{\text{inter}} &\leq \sum_{p= R-\kappa}^{R-2} \frac{D_{\text{max}} {\Delta}_{\text{max}}}{p(E-1)} \leq \frac{\left(\kappa-1\right) D_{\text{max}} {\Delta}_{\text{max}}}{\left(R-\kappa\right) \left(E-1\right)}
\end{align}

Similarly, we could obtain 

\begin{align}
\label{equ45}
\epsilon_{r}^{\text{inter}} 
% &\leq \sum_{p= R-\kappa}^{r-1} \frac{D_{\text{max}} \bar{\Delta}_{\text{max}}}{p(E-1)} \notag\\
\leq \frac{\left(r-R+\kappa\right) D_{\text{max}} \bar{\Delta}_{\text{max}}}{\left(R-\kappa\right) \left(E-1\right)}.
\end{align}

By substitute \eqref{equ45} into the \(\epsilon^{\text{inter}}\) term in \eqref{equ34}, we have
\begin{align}
\label{equ46}
\epsilon^{\text{inter}} 
% &\leq \frac{D_{\text{max}} \bar{\Delta}_{\text{max}} \prod_{i=R-\kappa+2}^{R-1} D_{i}}{(R-\kappa)(E-1)}  \left(\sum_{i=1}^{\kappa-1}i\right) \notag\\
&\leq\frac{\kappa(\kappa-1)D_{\text{max}} \bar{\Delta}_{\text{max}} \prod_{i=R-\kappa+2}^{R-1} D_{i}}{2(R-\kappa)(E-1)}.
\end{align}

Finally, we will give the upper bound of \(A2\). \(A2\) is the model divergence between the cell-centralized SGD algorithm and \(\bm{w}^{*}\). If we see each cell as a client, the cell-centralized SGD is the standard Fedavg algorithm. Under the assumption that \(\bm{w}^{*}\) is the final model of global centralized SGD, where data in all clients are collected in the CS to train a global model, \cite{a34} gives the bound of \(A_2\) as 
\vspace{-0.1cm}
\begin{align}
\label{equ47}
A_2 &= \left\| \bm{w}_{R}^{(c)} - \bm{w}^{*} \right\| \notag\\
&\leq \frac{\sum_{j=1}^{3} \hat{N}^{(f_j)} H^{(j)}\left( \sum_{i=1}^{C} \left\| P_{y=i}^{(c_j)} - P_{y=i}^{(c)} \right\|\right) }{N(R-1)(E-1)}. 
\end{align}
where \(H^{(j)}= \sum_{e=0}^{E-1} \left(\prod_{d=e+1}^{E-1}a_{R-1,d}^{(c_j)} \right) g_{\text{max}}( \bm{w}_{R-1,e}^{(c)})\).

Substituting~\eqref{equ37},~\eqref{equ46} and~\eqref{equ47} into \eqref{equ35} completes the proof.
\vspace{-0.3cm} 
\bibliographystyle{IEEEtran}  % 使用 IEEE 标准格式
\bibliography{references}      % 指定 BibTeX 文件（文件名不带 .bib）

% \begin{thebibliography}{1}
% \bibliographystyle{IEEEtran}

% \bibitem{ref1}
% {\it{Mathematics Into Type}}. American Mathematical Society. [Online]. Available: https://www.ams.org/arc/styleguide/mit-2.pdf

% \bibitem{ref2}
% T. W. Chaundy, P. R. Barrett and C. Batey, {\it{The Printing of Mathematics}}. London, U.K., Oxford Univ. Press, 1954.

% \bibitem{ref3}
% F. Mittelbach and M. Goossens, {\it{The \LaTeX Companion}}, 2nd ed. Boston, MA, USA: Pearson, 2004.

% \bibitem{ref4}
% G. Gr\"atzer, {\it{More Math Into LaTeX}}, New York, NY, USA: Springer, 2007.

% \bibitem{ref5}M. Letourneau and J. W. Sharp, {\it{AMS-StyleGuide-online.pdf,}} American Mathematical Society, Providence, RI, USA, [Online]. Available: http://www.ams.org/arc/styleguide/index.html

% \bibitem{ref6}
% H. Sira-Ramirez, ``On the sliding mode control of nonlinear systems,'' \textit{Syst. Control Lett.}, vol. 19, pp. 303--312, 1992.

% \bibitem{ref7}
% A. Levant, ``Exact differentiation of signals with unbounded higher derivatives,''  in \textit{Proc. 45th IEEE Conf. Decis.
% Control}, San Diego, CA, USA, 2006, pp. 5585--5590. DOI: 10.1109/CDC.2006.377165.

% \bibitem{ref8}
% M. Fliess, C. Join, and H. Sira-Ramirez, ``Non-linear estimation is easy,'' \textit{Int. J. Model., Ident. Control}, vol. 4, no. 1, pp. 12--27, 2008.

% \bibitem{ref9}
% R. Ortega, A. Astolfi, G. Bastin, and H. Rodriguez, ``Stabilization of food-chain systems using a port-controlled Hamiltonian description,'' in \textit{Proc. Amer. Control Conf.}, Chicago, IL, USA,
% 2000, pp. 2245--2249.

% \end{thebibliography}

\end{document}